\documentclass[pra,twocolumn,superscriptaddress,notitlepage,longbibliography,showkeys,showpacs,floatfix]{revtex4-1}

\usepackage[toc,page]{appendix}
\usepackage{bm, hyperref}
\usepackage{graphics}
\usepackage{amsmath,amsthm,amsfonts,amssymb,physics}
\usepackage{epsfig}
\usepackage{graphicx}
\usepackage[capitalise]{cleveref}
\newcommand{\argmin}{\operatornamewithlimits{argmin}}
\newcommand{\argmax}{\operatornamewithlimits{argmax}}
\newcommand{\spin}{{\bf S}}
\newcommand{\Span}{{\rm span}}
\newcommand{\hilbert}{\mathbb{H}}
\newcommand{\hilbertTensor}{\mathbb{H}^{\otimes}}
\newcommand{\compatstates}[2]{\mathfrak{C}_{#1}(#2)}

\usepackage[dvipsnames]{xcolor}

\newtheorem{prop}{Proposition}[section]
\newtheorem{lem}[prop]{Lemma}
\newtheorem{cor}[prop]{Corollary}
\newtheorem{obs}[prop]{Observation}

\begin{document}
\title{Quantum Covariance Scalar Products and Efficient Estimation of Max-Ent Projections}

\author{F. T. B. Pérez}
\author{J. M. Matera}
\affiliation{IFLP - CONICET, Departamento de Física, Facultad de Ciencias Exactas, Universidad Nacional de La Plata, C.C. 67, 1900 La Plata, Argentina}
\date{\today}
\begin{abstract}
The maximum-entropy principle (Max-Ent) is a valuable and extensively used tool in statistical mechanics and quantum information theory. It provides a method for inferring the state of a system by utilizing a reduced set of parameters associated with measurable quantities. However, the computational cost of employing Max-Ent projections in simulations of quantum many-body systems is a significant drawback, primarily due to the computational cost of evaluating these projections.
In this work, a different approach for estimating Max-Ent projections is proposed. The approach involves replacing the expensive Max-Ent induced local geometry, represented by the Kubo-Mori-Bogoliubov (KMB) scalar product, with a less computationally demanding geometry. Specifically, a new local geometry is defined in terms of the quantum analog of the covariance scalar product for classical random variables. Relations between induced distances and projections for both products are explored. Connections with standard variational and dynamical Mean-Field approaches are discussed.
The effectiveness of the approach is calibrated and illustrated by its application to the dynamic of excitations in a XX Heisenberg spin-$\frac{1}{2}$ chain model.
\end{abstract}
\maketitle
The field of Quantum Simulation in physics has gained significant attention in recent years \cite{feynman_simulating_1982, shor_algorithms_1994, nielsen_quantum_2000} due to its profound implications for the study of efficient and reliable control of large-scale quantum many-body systems. Consequently, the development of efficient simulation techniques for quantum many-body systems has become closely intertwined with this interdisciplinary domain, as understanding the evolution and dynamical properties of such systems is crucial for effective control strategies \cite{mansell_cold-atoms_2014, gordon_universal_2007}.

Nonetheless, studying the exact dynamics of open and closed quantum many-body systems remains a fundamental challenge in quantum mechanics \cite{nielsen_quantum_2000, shor_algorithms_1994, Hi.2013, HS.2017}. The main obstacle in solving exact dynamics lies in the coupling of the equations governing the evolution of expectation values of the observables of interest, encompassing all possible $n$-body correlations.

The dynamics of non-interacting systems, as well as the so-called \emph{Gaussian dynamics} (i.e. free quantum field theories) are exceptions, as the former preserves product states, while in the latter case, any correlation is a function of the pairwise correlations, giving rise to the famous Martin-Schwinger hierarchy of $n$-body correlations \cite{bruus2004many}. For this reason, non-perturbative schemes like the Mean Field Theory (MFT), in its different flavors and variants \cite{Auerbach.1994, Ring.2005, MRC.10, BRCM.15}, offer (a family of) prescriptions for building approximate and analytically solvable dynamics and equilibrium states, by exploiting the features of these types of tractable systems. 

%The key idea is to approximate the instantaneous state of the system through its least-biased estimation, known as the Max-Ent state, conditioned on the known expectation values of a \textit{small} set of 'relevant' observables, $B$. 

An extension of these non-perturbative methods was proposed by Balian et al. \cite{BAR.86}, based on Jaynes's Max-Ent principle \cite{Jaynes1, Jaynes2, Balian.1991}. 

The Max-Ent principle posits that given knowledge about the expectation values of a certain reduced set of \emph{relevant}/\emph{accesible} observables, the state of the system is the one with maximum (von Neumann's) entropy, consistent with the known expectation values. 
Each set of (linearly independent) relevant observables defines, therefore, a family of Max-Ent states, continuously parametrized by the expectation values. 
Then, the idea is to approximate the exact dynamics by a Max-Ent dynamics, i.e. a restricted dynamics over the family of the Max-Ent states.
In a similar way to the Nakajima-Zwanzig (NZ) projection technique \cite{heinz_theory_2002,Nakajima.58,Mori.65,IM.19arxiv},
the effective equations of motion are obtained via a linearized projection over the original ones.
However, unlike the NZ formulation, it does not necessarily rely on the division system/environment, making it more versatile. Also, 
unlike perturbative expansions, the accuracy of this approximation not only depends on the number of terms in the expansion but also on the choice of the relevant observables.
Nevertheless, the main challenge in this approach lies in the implementation of the constraint itself: the projection is expressed in terms of an orthogonal expansion with respect to the state-dependent Kubo-Mori-Bogoliubov (KMB) scalar product\cite{petz1993bogoliubov}. This scalar product depends on the spectral decomposition of the state that defines it, making its evaluation computationally very expensive. This
limits the applicability of the approximation to cases covered by the Time-Dependent Mean Field Theory (TDMFT) approaches, i.e. where separable or Gaussian states are assumed.

In this work, we delve into an alternative method for implementing the instantaneous projection, which enables the efficient solution of restricted dynamics for a wider range of sets of relevant observables. 
To accomplish this, we critically reexamine the key properties of the KMB scalar product and explore other computationally less demanding scalar products (and their induced geometries). 
Specifically, we focus on the quantum generalization of the \emph{covariance} scalar product \, \emph{covar}, which exhibits similar metric properties to the KMB scalar product, while, at the same time, being computationally less demanding. 

Furthermore, it is worth mentioning that both the KMB scalar product, as well as the solutions of the restricted dynamics, play a role in several branches of physics, including transport theory, linear response theory, the Kondo problem, non-commutative probability theory, and condensed matter physics \cite{naudts1975linear,petz1993bogoliubov,tanaka.06,scandi2019}, making the development of computable bounds and approximations a valuable tool in these areas.

The work is organized as follows. In the first section, we provide a brief review of the Max-Ent restricted dynamic formalism.
Then, in \cref{sec:computable}, we thoroughly reexamine the properties of the KMB and \, \emph{covar} scalar products. 
\cref{sec:test_ex} presents a detailed comparison between the exact dynamics and the approaches discussed. 
Finally, in  \cref{sec:discussion}, we present a general discussion of the results and perspectives. The Appendix contains proofs of the statements and mathematical details.

\section{Max-Ent Dynamics}

In this section, a review of the main concepts and results of the theory of Max-Ent states \cite{Jaynes1, Balian.1991} and Max-Ent restricted dynamics \cite{BAR.86} is presented. 

\subsection{Max-Ent Principle}\label{section_max_ent}

\paragraph{The Max-Ent principle in Quantum Mechanics}\smallbreak{}

Consider a quantum many-body system, with Hamiltonian ${\bf H}$
and algebra of observables ${\cal A}$ acting on a Hilbert space $\hilbert$, with space of states
\begin{equation}
  {\mathcal S}(\hilbert)=\{\rho\;|\;\rho
  \in {\cal B}(\hilbert),\;
  \rho\geq 0, \;\Tr\rho=1\},
\end{equation}
\noindent with ${\cal B}({\hilbert})$ the set of bounded operators acting on
$\hilbert$~\cite{Hall2013,nielsen_quantum_2000}.
The expectation value of an observable ${\bf O}\in{\cal A}$ is, then, given by
\begin{equation}
  \label{eq:av}
\langle {\bf O}\rangle_\rho=\Tr(\rho {\bf O}).  
\end{equation}
Conversely, if $B=\{{\bf Q}_{1},\ldots, {\bf Q}_{\dim(\hilbert)^2-1}\}$\footnote{This basis has
  dimension $\dim(\hilbert)^2-1$ since the identity operator $\mathbf{id}_{\hilbert}$ is fixed due to the
  constraint $\Tr\rho=1$. Notice however that for some infinity dimensional algebras (like the bosonic algebra), in order to satisfy the closeness condition, ${\bf id}\in{\cal A}$.
  } is a complete set of linearly independent operators s.t.\ ${\cal A}_{B} \equiv \Span(B) = {\cal A}$,
then knowledge about the system is \emph{complete}. Therefore, $\rho$ is completely determined
by the values of $\langle {\bf Q}_{\alpha}\rangle$~\cite{nielsen_quantum_2000}.
However, in many-body systems, the dimension of the algebra ${\cal A}$ grows geometrically with the number of components, making it unfeasible to access the expectation values of even a small fraction of the observables in ${\cal A}$.
On the other hand, by choosing $B=\{{\bf Q}_{1},\ldots,{\bf Q}_N\}$, with $N<\dim(\hilbert)^2-1$, as an independent set of \emph{accessible} observables, the information of their expectation values does not specify a \emph{single} density operator but a convex subset of ${\cal S}(\hilbert)$:
\begin{equation}
  \compatstates{B}{\mu_\alpha} = \left \{ \sigma \; |  \;\sigma \in {\cal S}(\hilbert),
         \Tr \sigma {\bf Q}_{\alpha}=\mu_\alpha, %\\
         {\bf Q}_{\alpha} \in B
\right\}
\end{equation}
\noindent with $\mu_\alpha=\langle {\bf Q}_\alpha\rangle$ the values of the known expectation values.  
With respect to the operators ${\bf Q} \in B$, $\compatstates{B}{\mu_{\alpha}}$ is an equivalence class, meaning that all of its elements are physically and statistically equivalent. 
However, this is not true for other operators. 
In particular, the dynamics of the state -- and of its expectation values -- depends on $[{\bf H},{\bf Q}]/{\bf i}$, which in general is not in ${\cal A}_B$. 

Subsequently, a crucial question arises: when estimating the expectation value of any other observable ${\cal Q} \notin B$, which state $\sigma$ would provide the fairest and unbiased choice for $\rho$? 
As shown in subsequent sections of this article, this question holds particular significance in the context of the evolution of expectation values. \smallbreak

One possible answer, and the one which will be explored in this article, lies in the Maximum Entropy (Max-Ent) principle~\cite{Jaynes1,canosa1990MaxEnt}. 
This principle states that the optimal choice is the one that maximizes the von Neumann entropy~\cite{nielsen_quantum_2000}
\begin{equation}\label{eq:vnEntropy}
    S(\rho)=-\Tr \rho \log \rho, 
\end{equation}
\noindent over $\compatstates{B}{\mu_\alpha}$.
In other words,
\begin{equation}
    \sigma^\star = \argmax_{\sigma \in \compatstates{B}{\mu_\alpha}} S(\sigma). 
\end{equation}
Note that, if ${\cal A}_B={\cal A}$ (namely, the set is equal to its closure), $\sigma^\star=\rho$ is the unique element in $\compatstates{B}{\mu_\alpha}$. On the other hand, for a generic basis $B$,  $\sigma^\star$ represents (in some sense) an even statistical mixture of the states in $\compatstates{B}{\mu_\alpha}$. 
Due to the convexity of this set, the convexity of the von Neumann entropy, and the linearity of the map between states and expectation values, it is easy to verify that ${\cal P}_B(\rho)=\sigma^*$ defines a smooth non-linear projection map from ${\cal S}(\hilbert)$ onto a manifold of Max-Ent states ${\cal M}_B$ associated to $B$ ~\cite{walet_thermal_1990,jost2017riemannian,HB.2017arxiv}, i.e. 

\begin{equation*}
    \begin{split}
        {\cal P}_B:{\cal S}(\hilbert)&\rightarrow {\cal M}_B,\\
        & \textnormal{s.t. } {\cal P}_B({\cal P}_B (\rho))={\cal P}_B( \rho).
    \end{split}
\end{equation*}

Moreover, one appealing property of the Max-Ent states, and of the Max-Ent manifolds as well, is given by the following proposition,

\begin{prop}\label{prop:Max-Ent_are_gibbs}
  Let $\sigma^*$ be the Max-Ent state associated with the observables in $B$. 
  \begin{equation}
    \label{eq:gibbsform}
    \sigma^\star = {e^{-{\bf K}}},\;\;{\bf K}\in {\cal A}_B \equiv \textnormal{span}(B),
  \end{equation}
\end{prop}

\noindent where the ${\bf K}$ operator is chosen s.t. ${\Tr \sigma^\star} = 1$.
Hence, Max-Ent states are (quantum) Gibbs states for a system with an effective Hamiltonian, ${\bf K}/\beta$, given as a linear combination of the operators $B = \{{\bf Q}_{\alpha}\}$ with real coefficients ~\cite{kubo_statistical-mechanical_1957,kubo_statistical-mechanical_1957,canosa1990MaxEnt,rossignoli_spin-projected_1996,HHW.67}. 
A rigorous proof of this statement is given in \cref{subsec:mfgibbs}.
By making use of these results,  ${\cal S}(\hilbert)$ (${\cal M}_B$) can be thought of as the image of the exponential map onto ${\cal A}$ (${\cal A}_B$), which is a real vector (sub)space. Namely, 

\begin{equation*}
\begin{split}
    \exp: {\cal A} \rightarrow {\cal S}(\hilbert), \\
    \exp: {\cal A}_B \rightarrow {\cal M}_B.
\end{split}
\end{equation*}

Moreover, the projection ${\cal P}_B$ naturally induces a (non-linear) smooth projection operator $\Pi_B:{\cal A}\rightarrow {\cal A}_B$, s.t. 

\begin{equation}\label{eq:compatconddiffbis}
\exp(-\Pi_B {\bf K})=  {\cal P}_{B}\left(e^{-{\bf K}}\right) .
\end{equation}

Notice that ${\cal P}_B$ can also be characterized in terms of the quantum relative entropy 
\begin{equation}\label{eq:RE}
S(\rho\|\sigma)=\Tr\left[\rho (\log \rho-\log \sigma)\right],
\end{equation}
as 
\begin{equation}\label{eq:REPDef}
{\cal P}_B(\rho) =\argmin_{\sigma \in {\cal M}_B} S(\rho\|\sigma).
\end{equation}

\noindent characterizing $\sigma$ as the less statistically distinguishable state from $\rho$, provided a fixed number of copies of the state are accessible, and for any observable \cite{Vedral.02}.

\smallbreak

\paragraph{Example.} Consider the state space of a single spin-$\frac{1}{2}$ system, which is the well-known Bloch sphere.
The full algebra of observables ${\cal A}$ is generated by ${B}={{\bf S}_{\mu=x,y,z}}$, representing the spin components along the Cartesian axes. If the set of accessible observables is restricted, for example, to only $B=\{{\bf S}_x, {\bf S}_y\}$, then $\compatstates{B}{\mu_{\alpha=x,y}}$ corresponds to the intersection of a straight line, parallel to the $z$-axis and containing $\rho$, with the Bloch sphere. This intersection yields the Max-Ent manifold as the intersection of the Bloch sphere with the $xy$-plane, as depicted in \cref{fig:Max-Ent1}.\smallbreak{}

\begin{figure}[htbp]
    \centering
    \includegraphics[width=\linewidth]{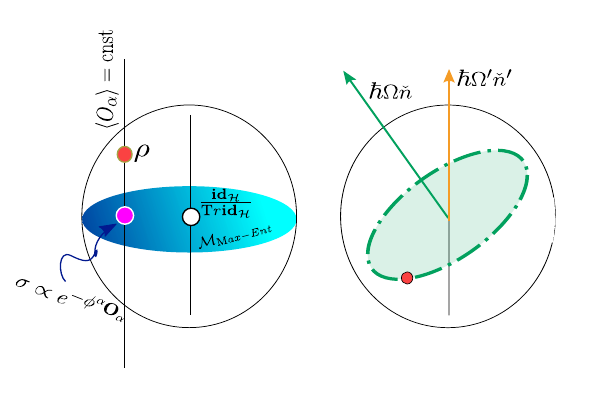}
    \caption{\label{fig:Max-Ent1}~Max-Ent dynamics in the Bloch's sphere. Left: Max-Ent construction. The Max-Ent manifold ${\cal M}_{\textnormal{Max-Ent}}$ spans all the states with defined $\langle{\bf Q}_{x,y}\rangle = \langle {\bf S}_{x,y}\rangle$ mean values and maximum entropy. The state $\sigma$ is the state with maximum entropy that shares these averages with $\rho$.
    Right: the ideas of exact and Max-Ent dynamics are contrasted, where the latter is an approximation of the former.
    }
\end{figure} 

\paragraph{Comment.} It is important to note that while $\compatstates{B}{\mu_\alpha}$ is a convex set, this is not necessarily true for ${\cal M}_{B}$ in general. 
In particular, if $S(\hilbert)$ represents the state space of a Spin-1 system and $B=\{{\bf S}_x,{\bf S}_y\}$, the states
$$
    \rho_{\mu=x,y}=\exp(-\ln(3){\bf S}_{\mu})/\Tr \exp(-\ln(3){\bf S}_{\mu}),
$$
belong to ${\cal M}_{B}$, while

$$
\rho=\frac{1}{2}\rho_x+\frac{1}{2}\rho_y\propto e^{-\lambda ({\bf S}_x+{\bf S}_y)-\kappa \{({\bf S}_x,{\bf S}_y)\}-\xi {\bf S}_z{\bf S}_z},
$$

\noindent with $\lambda\approx -0.526$, $\kappa\approx 0.137$, and $\xi\approx 0.125$ does not belong to ${\cal M}_{B}$, as it cannot be written in the form of \cref{eq:gibbsform}.\smallbreak

\subsection{Linear Projections and Geometry of  ${\cal M}_B$}\label{subsec:geometry}

Evaluating ${\cal P}_B$ is, however, a challenging optimization problem, primarily due to the high computational cost associated with the exact evaluation of $\rho$. 
Nonetheless, in certain cases, our interest lies in projecting states $\rho = \exp(-{\bf K})$ onto the neighborhood of a specific $\rho_0 =  \exp(-{\bf K}_0)\in {\cal M}_B$.
In such scenarios, it becomes reasonable to approximate $\Pi_{B}({\bf K})$ using a linear projector $\pi\equiv\pi_{B,\rho_0}$\cite{Amari2000,walet_thermal_1990,IM.19arxiv}, 

\begin{equation}\label{eq:compatconddiff}
    \pi_{B,\rho_0} \left(\Delta {\bf K}\right)\equiv  \left.\frac{\partial}{\partial \lambda}\Pi_B\left({\bf K}_0 + \lambda \Delta {\bf K}\right)\right|_{\lambda\rightarrow 0},
\end{equation}
\noindent for any $\Delta {\bf K}\in {\cal A}$. 
In order to explicitly construct $\pi_{B,\rho_0}$, observe that

\begin{equation}
  \label{eq:nlProj}
\langle {\bf O}\rangle_{{\cal P}(\rho)}=\langle {\bf O}\rangle_{\rho}, \quad \forall {\bf O} \in{\cal A}_B.
\end{equation}

Assuming that $\rho \approx {\cal P}(\rho) \rightarrow {\bf K}\approx \Pi({\bf K})$ (i.e., the exact state of the system lies close to the Max-Ent manifold), \cref{eq:nlProj} can be linearized  around ${\cal P}(\rho)$ using the following property:

\begin{prop}\label{prop:kubo1}
  Let $\rho_\lambda = \exp(-{\bf K}_0 + \lambda \Delta{\bf K})$ with ${\bf K}_0={\bf K}_0^\dagger\in {\cal A}$. 
  Then,
\begin{equation}\label{eq:KWBvar}
    \frac{\partial}{\partial \lambda} \Tr \rho_\lambda {\bf O}  
    = (\Delta{\bf K}^\dagger, {\bf O})_{\rho_\lambda}^{\textnormal{KMB}},
\end{equation}

\noindent with

\begin{equation}\label{eq:KWBdef1}
    ({\bf Q}_1,{\bf Q}_2)_{\sigma}^\textnormal{KMB}
        =
    \int_0^1  \Tr[\sigma^{1-\tau}{\bf Q}_1^\dagger \sigma^\tau {\bf Q}_2] d\tau.
\end{equation}

\noindent the Kubo-Mori-Bogoliubov (KMB) scalar product~\cite{petz1993bogoliubov} relative to $\sigma$.
\end{prop}
The proof of this proposition can be found in  \cite{BAR.86,petz1993bogoliubov} and is included in \cref{app:proof_kubosp} for completeness. \smallbreak
\cref{prop:kubo1} allows for the characterization of $\pi_{B,\sigma}$ as an orthogonal projector wrt. the KMB product.

\begin{prop}\label{prop:ortproj}
\cref{eq:compatconddiff} is satisfied if $\pi_{B,\rho_0}$ is an \emph{orthogonal projection} with respect to $(\cdot,\cdot)^\textnormal{KMB}_{\rho_0}$, meaning that for all states $\rho_0 \in {\cal M}_B$:

\begin{equation}\label{eq:ortprojdef}
    \left({\bf Q},\pi_{B,\rho_0} {\bf O}\right)_{\rho_0}^{\textnormal{KMB}}=({\bf Q},{\bf O})_{\rho_0}^{\textnormal{KMB}}\; \quad \begin{array}{cc}
         \forall{\bf Q} \in {\cal A}_B \\
         {\bf O}\in {\cal A}. \\
    \end{array}
\end{equation} 
\end{prop}

A proof of this proposition can be found in \cref{sec:proofOrtProj}. \smallbreak

\cref{prop:ortproj} is very important for several reasons. 
The first one is practical because it allows for the explicit computation of the projection, $\pi_{B,\sigma} \,{\bf K}$, in terms of operators ${\bf Q}_\alpha\in B$ as a Bessel-Fourier expansion

\begin{equation}\label{eq:piasbesselfourier}
\pi_{B,\sigma}\,{\bf K}=\sum_{\alpha\beta} [({\cal G}_{B,\sigma}^{\textnormal{KMB}})^{-1}]^{\alpha\beta} ({\bf Q}_\alpha,{\bf K})^\textnormal{KMB}_{\sigma}{\bf Q}_\beta,
\end{equation}

with 

\begin{equation}\label{eq:coeffsGramm}
    [{\cal G}_{B,\sigma}^\textnormal{KMB}]_{\alpha\beta}=({\bf Q}_\alpha, {\bf Q}_\beta)^\textnormal{KMB}_{\sigma}, 
\end{equation}

\noindent the \emph{Gram's matrix} of the basis of accessible observables w.r.t the KMB scalar product.

On the other hand, \cref{prop:ortproj} provides a way to reformulate \cref{eq:compatconddiff} in geometrical terms, yielding very fruitful results in the way of bounds, approximations, and other metric properties.
For example, we notice that \cref{prop:ortproj} implies that

\begin{equation}\label{eq:linprojMinCondition}
    \pi_{B,\rho}({\bf K}) = \argmin_{{\bf K}'\in {\bf B}}\|{\bf K}'-{\bf K}\|^{ \rm KMB}_{\rho}, \,\, \rho = \exp(-{\bf K}),
\end{equation}

\noindent  with 

\begin{equation}\label{eq:inducedKMBmetric}
    \|{\bf A}\|_{\rho_0}^{\rm KMB}=\sqrt{({\bf A},{\bf A})_{\rho_0}^{\rm KMB}},
\end{equation}

\noindent the KMB \emph{induced distance} in the neighborhood of $\rho_0$. This metric is closely related to the relative entropy \cref{eq:RE} between states in the neighborhood of $\rho_0$, which it bounds, through the exponential map parametrization (see \cref{prop:kmbnorm_and_re} in the Appendix for further discussion on this subject).

\subsection{Projected Dynamics and Restricted Schr\"odinger Dynamics }\label{sec:projected_and_constrained_dynamics}
Until now, we have considered the Max-Ent projection (and its linearization) for an instantaneous state.
Let's consider now a system whose state was initially described by $\rho_0 = \rho(0)\in{\cal M}_B$ that undergoes a closed evolution governed by the Schrödinger equation:
\begin{equation} \label{eq:Schrodinger}
    {\bf i} \hbar \frac{d\rho}{dt} = [{\bf H}, \rho],
\end{equation}

\noindent where ${\bf H} \in {\cal A}$ represents the system's Hamiltonian.
For the present developments, it is convenient to work with the dynamics of ${\bf K}(t)=-\log(\rho(t))$, through the following

\begin{figure}    
\centering
\includegraphics[width=.9\linewidth]{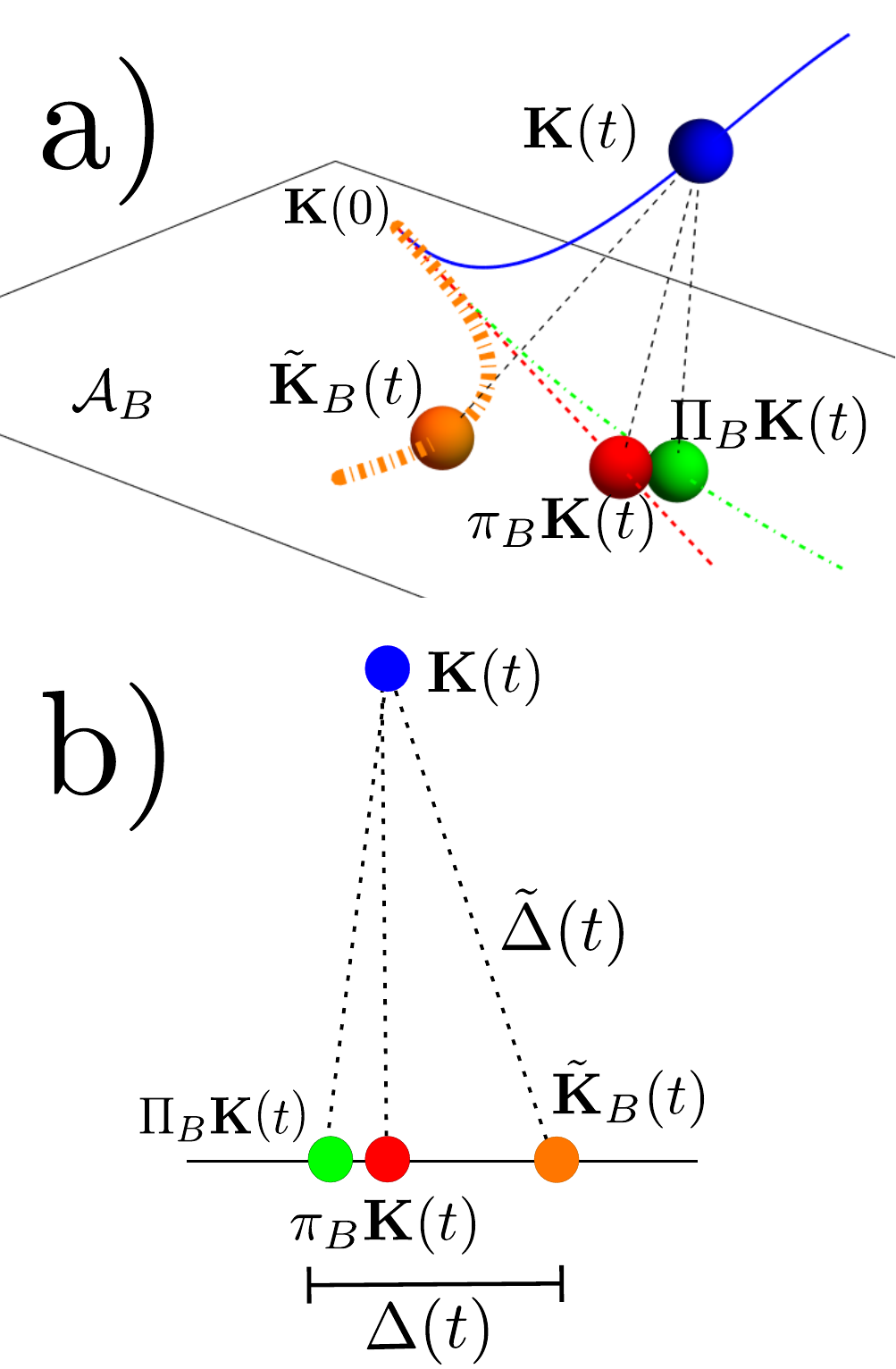}
\caption{\label{fig:trajectories}(Color online) Different evolution schemes. a) The solid curve (blue online) represents the trajectory of ${\bf K}(t)$ following the free Schr\"odinger evolution. 
The dot-dashed line (green online) curve and the dashed line (red online) represent the Max-Ent projection of the free evolution $\Pi_B {\bf K}(t)$  and its linearization $\pi_B {\bf K}(t)$, respectively. 
The dotted line (orange online) represents the restricted evolution $\tilde{\bf K}_B(t)$ \cref{eq:restricted}.
b) Relation among the distances $\Delta(t)$ and $\tilde\Delta(t)$ \cref{eq:deftildeDelta,eq:defDelta} %Eqs.~(\ref{eq:deftildeDelta})~and~(\ref{eq:defDelta}) 
between the instantaneous ${\bf K}(t)$, $\Pi_B{\bf K}(t)$ and $\tilde{\bf K}_B(t)$ and 
its different approximations. In the scheme, intrinsic KMB geometry around 
$\pi_{B,\rho(t)} {\bf K}(t)$ is identified with the Euclidean one.  
Note that the different states do not lie in the same trajectory. 
}
\end{figure}

\begin{lem}\label{lem:k_scr_dyn}
Let  ${\bf K}(t)=-\log(\rho(t))$, with $\rho(t)$ a solution of \cref{eq:Schrodinger} for a certain Hamiltonian ${\bf H}$. Then,

\begin{equation} \label{eq:SchrodingerK}
    {\bf i} \hbar \frac{d{\bf K}}{dt} = [{\bf H}, {\bf K}].
\end{equation}
\end{lem}
The proof of this Lemma is shown in \cref{sec:proofSchK}. 

Given that the accessible observables are limited to those in ${\cal A}_{B} = \Span(B)$, it is meaningful to examine the evolution of the projection:

\begin{equation}\label{eq:projecteddynamicsbis}
    \Pi_B ({\bf K}(t))= - \log( {\cal P}_B (\rho(t)) ),
\end{equation}

\noindent which offers a more concise representation of the state with respect to the accessible observables. 
Evaluating $\Pi_B ({\bf K}(t))$ is, however, problematic since it necessitates solving the full Schr\"odinger equation, \cref{eq:Schrodinger}, and subsequently computing the projection itself.
Instead, in the neighborhood of the Max-Ent manifold ${\cal M}_B$, or equivalently of ${\cal A}_B$, the linearization of the Max-Ent projection $\Pi_B$ yields an orthogonal projection w.r.t. the KMB geometry $\pi_B$, as shown in the preceding sections. 
Thus, it makes sense to study the following \underline{\textit{(linear) projected evolution}}: 

\begin{equation}\label{eq:projecteddynamics}
    \pi_{B,\rho(t)} {\bf K}(t) = - \pi_{B,\rho(t)} \log(\rho(t))\,.
\end{equation}

\noindent Note that evaluation of \cref{eq:projecteddynamics} still necessitates solving the full Schr\"odinger equation and computing a (now-linear) projection. 
For many-body systems, solving the Schr\"odinger equation is not possible,
undermining the feasibility of employing a projection approximation.
Nevertheless, by assuming that $\rho(t)$ evolves in the neighborhood of ${\cal M}_B$, $\pi_B {\bf K}(t)$ can be approximated by a \underline{\emph{restricted} dynamics} ${\tilde{\bf K}_B}(t)$.
If $\tilde\sigma(t) = \exp \big(-\tilde{{\bf K}}_B(t) \big)$, then 

\begin{equation}\label{eq:restricted}
    {\bf i} \hbar \frac{d \tilde{\bf K}_B}{dt} = \pi_{B,\tilde\sigma(t)}\big([{\bf H}, \tilde{\bf K}_B]\big),
\end{equation}
represents a Schr\"odinger evolution restricted to the Max-Ent manifold ${\cal M}_B$ \cite{hackl_etal, Amari2000,walet_thermal_1990,IM.19arxiv}.

By definition, \cref{eq:restricted} is a closed evolution on ${\cal A}_B$, since $\pi_B$ acts trivially on it.
Moreover, \cref{eq:restricted} is a non-linear differential equation (since the projection itself is calculated w.r.t. the KMB inner product) but local in time and has a formal solution for ${\tilde{\bf K}_B}(t)$ (more on this below). 
\smallbreak
The relation between these operators is depicted in panel a) of \cref{fig:trajectories}. 
As the system evolves, the Hamiltonian evolution pulls ${\bf K}(t)$ out of the relevant subspace, ${\cal A}_B$. 
The Max-Ent projection $\Pi_B {\bf K}(t)$ follows a trajectory, over ${\cal A}_B$, of (minus the logarithm of) Max-Ent states, sharing the same instantaneous expectation values with the free evolution.
The linearized projection $\pi_{B,\sigma(t)}{\bf K}(t)$ provides an approximation for $\Pi_B{\bf K}(t)$, valid provided ${\bf K}(t)$ remains close to ${\cal A}_B$.

As far as $\tilde{\bf K}_B$ is a good approximation to ${\bf K}$, then $\tilde{\bf K}_B(t)\approx \pi_{B,\rho(t)}{{\bf K}}(t) \approx \Pi_B{{\bf K}}(t)$. 
As shown below, this condition can be achieved by expanding the relevant set of observables through a judicious choice of operators (see \cref{sec:apphierarch}), for short enough times. 
For longer times, if ${\bf K}(t)$ stays close to ${\cal A}_B$, $\Pi_B{\bf K}(t)\approx\pi_{B,\rho(t)}{\bf{K}}(t)\approx {\bf K}(t)$, but $\tilde{\bf K}_B(t)$, due to the accumulated differences in the derivatives, eventually moves away from ${\bf K}$. 
However, as detailed in \cref{prop:entropy_and_normalization_preserving}, given a sensible choice of $B$, the restricted evolution $\tilde{\bf K}_B$, like with $\Pi_B{\bf K}$, conserves both the normalization and the relevant constants of motion, even if ${\bf K}(t)$ moves away from ${\cal A}_B$. 
Also, even if the instantaneous states diverge, the qualitative behavior of the orbits can remain similar. 

\paragraph{Example}

To illustrate how this approach works, let's revisit the dynamics in the Bloch sphere. Consider $B=\{{\bf S}_x,{\bf S}_y\}$, and let the Hamiltonian have the form ${\bf H}= \check{n}\cdot \vec{\bf S}=\Omega {\bf S}_z + \omega {\bf S}_x$ (see panel \cref{fig:Max-Ent1} b).
The exact dynamics describe circular trajectories around $\check{n}$ with an angular frequency of $\sqrt{\Omega^2+\omega^2}$. In this specific case, the natural projection over ${\cal M}_B$ (\cref{eq:nlProj}) coincides with the Euclidean projection onto the $x$-$y$ plane.
If $\omega=0$, the projection of the exact trajectory onto ${\cal M}_B$ coincides with the trajectory of the restricted dynamics. 
However, if $\omega\neq 0$, the restricted dynamics would still be a circular trajectory with an angular frequency of $\Omega$, while the projection of the exact dynamics would result in an elliptical trajectory with an angular frequency of $\sqrt{\Omega^2+\omega^2}$. When $|\omega/\Omega|\ll 1$, the restricted dynamics would closely approximate the projection of the exact dynamics. 
On the other hand, if we would choose the basis $B'=\{ \log \rho(0), -{\bf i}[{\bf H}, \log \rho(0)] \}$ instead of $B$, the dynamics would be always exact, despite $B'$ being non-complete, since ${\cal A}_{B'}$ is closed under the action of $[{\bf H},\cdot]/{\bf i} \hbar$). \smallbreak

\subsection{Explicit computation of the restricted dynamics}\label{sec:restricteddynamics}

Using the orthogonal expansion \cref{eq:piasbesselfourier} regarding a \emph{fixed} basis $B$, \cref{eq:restricted} can be expressed as a set of differential equations for the expansion coefficients $\phi_\mu(t)$ (where $\tilde{\bf K}_B (t)=\sum_\mu \phi_\mu(t){\bf Q}_\mu$):

\begin{equation}\label{eq:coeffs}
    \sum_\beta [{\cal G}_{B,\sigma(t)}^\textnormal{KMB}]_{\alpha\beta}\frac{d \phi_{\beta}}{dt}=\sum_{\beta}{\cal H}_{\alpha \beta}\phi_{\beta}(t),\
\end{equation}

\noindent with 
\begin{equation}\label{eq:coeffsH}
{\cal H}_{\alpha \beta}\equiv[{\cal H}^\textnormal{KMB}_{\sigma(t)}]_{\alpha\beta}=\frac{1}{{\bf i}\hbar}({\bf Q}_\alpha,[{\bf H}, {\bf Q}_\beta])^\textnormal{KMB}_{\sigma(t)}, 
\end{equation}
representing a real matrix that governs the dynamics of the coefficients. It is important to note that both ${\cal H}_{\alpha\beta}$ and ${\cal G}_{\alpha\beta}$ are non-linear functions of the instantaneous state ${\bf K}$, as $(\cdot,\cdot)_B^\textnormal{KMB}$ depends on $\exp\big(-\tilde{\bf K}_B(t)\big)$. Consequently, \cref{eq:coeffs} becomes a set of \emph{non-linear} coupled equations.
\smallbreak

\paragraph{Convergency}  
In the previous analysis, it was assumed that $\rho(t)\approx {\cal P}(\rho(t))$, in order to approximate the non-linear projector $\Pi_B$ by its linear approximation $\pi_{B,\sigma}$, and the projected dynamics -- \cref{eq:projecteddynamics} -- by the restricted dynamics -- \cref{eq:restricted}. Let's discuss now more carefully how these conditions are quantified. 

$\Pi_B$ and $\pi_{B,{\sigma}}$ are defined in terms of the minimization of two functionals, the relative entropy and the KMB distance, respectively. 
According to  \cref{prop:kmbnorm_and_re}, as long as the second-order expansion is valid, both quantities are monotones of each other, in a way that
$\|\Pi_B{\bf K}-{\bf K}\|_{\rho(t)}$ and $\|\pi_{B,\rho(t)}{\bf K}-{\bf K}\|_{\rho(t)}$ are equal up to a higher order in $\|\pi_{B,\rho(t)}{\bf K}-{\bf K}\|_{\rho(t)}$. On the other hand, $\tilde{\bf K}_B(t)$ is a solution of \cref{eq:restricted} s.t. $\tilde{\bf K}_B(0)=\Pi_B {\bf K}(0)=\pi_{B,\sigma(0)} {\bf K}(0)$, so for short times, 
$\tilde{\bf K}_B(t)\approx \Pi_B {\bf K}(t)$. It is convenient then to introduce 

\begin{equation}\label{eq:defDeltaK}
  \Delta{\bf K}(t)={\bf K}(t)-\tilde{\bf K}_B(t),
\end{equation}

\noindent as the difference between the free and the restricted evolutions. The KMB distance 

\begin{equation}\label{eq:deftildeDelta}
\tilde\Delta(t)=\|\Delta{\bf K}(t)\|_{\sigma(t)}^{\rm KMB},
\end{equation}

\noindent measures the effect of this difference in the estimation of expectation values and the distinguishability of the associated states. 
On the other hand, if we focus just on the relevant observables, what we look for is to approximate $\Pi_B{\bf K}(t)$, in a way that the figure of merit is

\begin{equation}\label{eq:defDelta}
\Delta(t)=\|\Pi_B\Delta{\bf K}(t)\|_{\sigma(t)}^{\rm KMB}.
\end{equation}

The relation between these quantities is depicted in the panel b) of \cref{fig:trajectories}.
Since we do not have direct access to ${\bf K}(t)$, we need an expression of $\Delta{\bf K}$ as a functional of $\tilde{\bf K}(t)$. 
From the results in \cref{app:sccriteria}, it follows that these quantities can be bounded, during the simulation, without an important overhead by

\begin{equation}\label{eq:error_bound}
\tilde\Delta(t)<\int_0^t \|[{\bf H},\tilde{\bf K}_B(t')]\|_{\sigma(t)}^{\rm KMB}dt'.
\end{equation}

In general, along the evolution, the system develops correlations not contained in $B$.
For example, in typical interacting many-body systems, an initially uncorrelated state develops ${\cal O}(t^n)$ non-trivial $n$-body correlations. 
Then, if $B$ contains just local observables, correlations are neglected in the evolution of $\tilde{\bf K}_B(t)$, while it does affect the dynamics of $\rho(t)$. 

Still, if these correlations do not affect the dynamics of the relevant expectation values in an appreciable manner, then $\exp\big(-\tilde{\bf K}_B(t)\big)$ provides a good approximation to ${\cal P}(\rho(t))$, even if it does not approximate correctly $\rho(t)$. 
On the other hand, if some correlations do heavily affect the dynamics of the relevant variables, those correlations can be seen as the actually \emph{relevant} observables, which can be inferred by looking at the dynamics of other observables.

Therefore, by extending the basis $B$, including these new relevant observables, it is possible to make the projected and restricted dynamics closer.   
This statement can be made mathematically precise by considering a sequence of \emph{Hierarchical bases} \cite{Iserles2005, saad2003iterative},
  
  \begin{equation}
    \label{eq:hierarchical_basis}
  B_0 \subset B_1 \subset B_2 \subset \ldots \subset B_{\ell} \subset \ldots, 
  \end{equation}

in a way that, by construction, 
$$
\|{\bf K}(t)-\pi_\ell {\bf K}(t)\|_{\rho(t)}^{\rm KMB}< \|
{\bf K}(t)-\pi_{\ell-1} {\bf K}(t)\|_{\rho(t)}^{\rm KMB},
$$
with $\pi_\ell\equiv \pi_{B_\ell}$ the orthogonal projector associated to the subspace ${\cal A}_\ell\equiv {\cal A}_{B_\ell}$ regarding the KMB scalar product relative to $\rho(t)$. 
These so-called Hierarchical bases are related to the Hierarchical Lie algebras, see \cite{Iserles2005}, which arise in the study of efficient solutions to differential equations on manifolds.
In our present case, however, the Hierarchical bases are not, in general, Lie algebras, only sharing an iterative commutator-based structure. 

Notice that, for finite dimensional algebras, the distance converges to $0$ for large enough $\ell$. 
However, the particular way in which the convergence is achieved depends strongly on the choice of ${B}_\ell$. The discussion of these conditions in a general context is out of the scope of this work. 
For the present analysis, we are going to focus on the case of dynamics generated by a time-independent Hamiltonian ${\bf H}$, and the sequence ${\bf b}_{\ell} \in B_\ell$ with

\begin{equation}
  \label{eq:iterated_comm_basis}
  {\bf b}_0={\bf K}(0)\;\;\mbox{and}\;\;{\bf b}_m=\frac{[{\bf H},{\bf b}_{m-1}]}{{\bf i}\hbar}
\end{equation}

As such, the subspaces ${\cal A}_\ell$, spanned by the Hierarchical bases, can be understood as Krylov subspaces generated by the initial operator ${\bf b}_0$ and the operator $\frac{1}{i\hbar} {\rm ad}_{{\bf H}}^{(\ell)}(\cdot)$.
In \cref{sec:apphierarch}, it is shown that the projected and the restricted dynamics are consistent with an $\ell-$th order perturbative expansion, and hence for a fixed $t_{max}>t$, the KMB distances (and any other metric) converge as $t_{max}^{\ell+1}$.  On the other hand, numerical simulations presented in \cref{sec:test_ex} seem to suggest that for larger times, the KMB distance reaches an asymptotic value, that decreases with $\ell$.

\section{Computable general Max-Ent Dynamics}\label{sec:computable}

With the method above, in principle, it is possible to solve the restricted dynamics for any choice of the physical system and set of relevant observables, involving just as many dynamical variables as the considered relevant independent observables. However, to explicitly solve the dynamics, the challenge lies in computing the self-consistent projections via the evaluation of the KMB scalar product of operators with respect to the instantaneous state $\sigma(t)$: its computation requires the construction and explicit diagonalization of the instantaneous state $\sigma(t)$ at each step of the evolution. This process can only be carried out explicitly for Gaussian and product states, and for very low-dimensional systems \cite{vidal_computable_2002}.

One way to overcome this limitation arises from the observation that the same projector can be orthogonal regarding distinct scalar products. Moreover, even if two scalar products lead to different but similar orthogonal projectors, choosing a suitable basis $B$, it can be expected that the dynamics induced by the projectors will be similar. 
In this section, the desired requisites for a computable generalization of the KMB dynamics are discussed in depth, alongside a concrete proposal fulfilling these requisites.
\smallbreak

\subsection{Required properties.} 

In the upcoming sections, an alternative proposition to solve the Max-Ent projected dynamics equation \cref{eq:coeffs} is to replace the KMB geometry with a mathematically similar yet computationally efficient geometry.
To this end, one must, first, embark on a search for an alternative scalar product $(\cdot,\cdot)$ that can serve as a replacement for the KMB scalar product while possessing comparable metric properties. 
A comprehensive analysis of the mathematical properties of this scalar product can be found in \cite{amorim2020complete}. 
Additionally, for a more extensive exploration of the broader applicability of this geometry, particularly from the perspective of operator theory, refer to the comprehensive summary provided in \cite{balazs2006hilbertschmidt}.
By pursuing this avenue, an improved approach for computing scalar products and orthogonalization of bases of observables, with higher computational efficiency, is desired.
In order to achieve results similar to those obtained using the KMB scalar product, the proposed alternative must satisfy several significant conditions.
\smallbreak

\paragraph{Reality condition.} Firstly, a suitable candidate for a scalar product must meet the reality condition(see \cref{app:realsp}),
\begin{equation}\label{eq:realcond} 
({\bf A}, {\bf B})^{'*}=({\bf A}^\dagger, {\bf B}^\dagger)'=({\bf B},{\bf A})'.
\end{equation}

This condition ensures that $\pi_B({\bf Q})=\pi_B({\bf Q})^\dagger \in {\cal A}_B$ for any ${\bf Q}={\bf Q}^\dagger \in{\cal A}$  and for any choice of $B$ s.t. ${\bf Q}\in {\cal A}_B \Rightarrow {\bf Q}^\dagger \in {\cal A}_B$, see \cref{app:realsp}. Both the KMB scalar product and the Hilbert-Schmidt scalar product (HS), given by 
$$({\bf A},{\bf B})^\textnormal{HS}= \Tr \mathbf{A}^\dagger{\bf B},$$
\noindent are real-valued scalar products~\cite{Hall2013}. 
\paragraph{Tensor-Product compatibility condition.} The HS scalar product is particularly advantageous as it is much easier to compute than the KMB scalar product when ${\bf A}^\dagger$ and ${\bf B}^\dagger$ represent $k$-body correlations. Furthermore, the HS scalar product is compatible with the tensor product operation:
\begin{equation}
\label{eq:compatSPTP}
({\bf O}_1 \otimes {\bf O}_2,{\bf Q}_1 \otimes {\bf Q}_2)_{\textnormal{HS}}=({\bf O}_1, {\bf Q}_1)({\bf O}_2 , {\bf Q}_2).
\end{equation}
This property is not shared by the KMB product, even if $\sigma(t)$ is a product operator, which makes the evaluation of $k$-body correlation functions much harder than in the HS geometry.

\paragraph{Statistical weight.} However, simple substitution of the KMB scalar product by the HS scalar product in \cref{eq:coeffs} is not always a viable approach. 
The KMB scalar product assigns weights to operators based on their statistical significance, while the HS scalar product is unitarily invariant. As a result, two operators that are close in terms of the KMB-induced norm may appear very different according to the HS-induced norm. 
This discrepancy arises, for example, when the operators differ in the form $|i\rangle \langle j|$, with $|i\rangle$ and $|j\rangle$ being states with very low occupation probabilities ($\langle i|\rho|i\rangle, \langle j|\rho|j\rangle\ll 1$). 
For instance, in a bosonic system where $\hat{\bf n}={\bf a}^\dagger{\bf a}$ is the number operator 
and $\rho$ is a Gaussian state with $\langle n\rangle \approx 1$, 
$|\hat{\bf n}^2-\hat{\bf n}|_\textnormal{KMB}=2\sqrt{13}\approx 7.21$, but $|\hat{\bf n}^2-\hat{\bf n}|_\textnormal{HS}$ is unbounded.
\smallbreak
% and exploring an alternative geometry for the space of Max-Ent states, given by a correlation scalar

\subsection{Quantum Covariance scalar product(\, \emph{covar}) and  \, \emph{covar} geometry}\label{subsec:geometrycov}

A more suitable choice of scalar product is given by the \emph{quantum COVARiance} scalar product w.r.t. a certain reference state $\sigma$ --from now on, \, \emph{covar}--,

\begin{equation}
\label{eq:corrSP}
({\bf O},{\bf Q})_{\sigma}^{{\rm covar}}=\rm Tr\left[\sigma \frac{\{{\bf O}^\dagger, {\bf Q}\}}{2}\right]
\end{equation}

\noindent which, for Hermitian inputs, is a real-valued scalar product.

This scalar product, up to a constant factor, reduces to the HS when $\sigma\propto \textnormal{id}_{\hilbert}$. On the other hand, for normalized reference states $\Tr\sigma=1$, it has a simple statistical interpretation: the scalar product of an operator with the identity operator yields its expectation value,
$$
    ({\bf id}_{\hilbert},{\bf Q})_{\sigma}^{{\rm covar}}=\langle {\bf Q} \rangle_{\sigma},
$$
while the scalar product between two operators with zero expectation value (i.e.\ orthogonal to the identity) is given by its covariance:
$$
{\rm Cov}_{\sigma}({\bf O},{\bf Q})=\left\langle \frac{\{{{\bf O},{\bf Q}}\} }{2}\right\rangle_{\sigma}-\langle {\bf O}\rangle \langle {\bf Q}\rangle_\sigma.
$$

Additionally, the induced norm for an operator with zero expectation value is given by its standard deviation:
$$
||{\bf Q}-\langle {\bf Q}\rangle||_{\sigma}^{{\rm covar}}=\sqrt{\langle {\bf Q}^2\rangle_\sigma -\langle {\bf Q}\rangle_\sigma^2}.
$$
Hence, the \, \emph{covar} scalar product can be regarded as the quantum analog of the \emph{covariance} scalar product between classical random variables.

Notably, the  \, \, \emph{covar} scalar product offers an advantage over the KMB geometry, as its computation does not require the diagonalized form of the reference state, making it more computationally efficient. Furthermore, as it is a linear function w.r.t. the reference state, it can be efficiently computed for any separable reference state $\rho_0=\sum_i q_i \rho_i^{\otimes}$.

Although it does not satisfy the tensor-product compatibility condition \cref{eq:corrSP}, for self-adjoint operators, it can be computed as the real part of the Gelfand-Naimark-Sigal (GNS) scalar product~\cite{Hall2013,CM.2020}
$$
    ({\bf O},{\bf Q})_{\sigma}^\textnormal{GNS}=\Tr[\sigma{\bf O}^\dagger {\bf Q}],
$$
\noindent which does satisfy it. 
For instance, choosing $\sigma=\bigotimes_i \sigma_{i}$, the scalar product between ${\bf O}=\bigotimes_i{\bf o}_i$ and ${\bf q}=\bigotimes_i{\bf q}_i$ is simply given by
$$
    ({\bf O},{\bf Q})_{\sigma}^{\rm covar}=\Re \left[\prod_i ({\bf o}_i,{\bf q}_i)_{\sigma_{i}} \right].
$$
Another important feature of the \, \, \emph{covar} scalar product is that, if one of the arguments commutes with the reference state, this product yields the same result that the KMB scalar product regarding the same reference state., i.e.

$$
[{\bf K},{\bf A}]=0 \Rightarrow  ({\bf A},{\bf B})_{\sigma}^{\rm KMB}=({\bf A},{\bf B})_{\sigma}^{\rm covar}.
$$

Therefore, \cref{prop:entropy_and_normalization_preserving} and \cref{prop:trivial_projection_in_expect_values} are also valid if the KMB product and orthogonal projectors are replaced by their corresponding \, \, \emph{covar} counterparts. As a result,
if instead of the KMB projector a \, \, \emph{covar} projector is used in \cref{eq:restricted}, both the KMB and the \, \, \emph{covar} trajectories lie over the same constant entropy submanifold of ${\cal M}_B$, and automatically preserve the normalization.

On the other hand, the  \, \emph{covar} geometry shares with KMB a common orthogonal basis of ${\cal A}$, with the norms of each vector related by a ${\cal O}(1)$ factor (see \cref{sec:KMBinducedNorm}):
\begin{eqnarray}
\big(|i\rangle\langle j|,|k\rangle\langle l|\big)_\sigma^\textnormal{KMB}&=& W_{ij} \times \big(|i\rangle\langle j|,|k\rangle\langle l|\big)_\sigma^{\rm covar},\\
W_{ij}&=&\frac{\tanh(\log(p_i/p_j)/2)}{\log(p_i/p_j)/2}\leq 1,
\end{eqnarray}
\noindent where $|i\rangle,|j\rangle,|k\rangle, |l\rangle$ are eigenvectors of $\sigma$ with eigenvalues $p_i,p_j,p_k,p_l$ respectively. As a result, both scalar products yield simular values for operators which connect states with similar probabilities.

The following proposition provides a useful tool to compare the induced geometries: 
\begin{prop}\label{lem:metricineq}
Let $\sigma \in {\cal S}(\hilbert)$ and ${\bf A}\in \cal A$. Then
\begin{equation}\label{eq:chain_norm_inequalities}
    \|{\bf A}\|\geq 
        ||{\bf A}||^{{\rm covar}}_{\sigma}   \geq 
        ||{\bf A}||^{\textnormal{KMB}}_{\sigma}  \geq  |S(\sigma\|e^{\log \sigma -{\bf A}})|\,.
 \end{equation}
 Notice that if $\Tr e^{\log \sigma -{\bf A}}=1$, then the absolute value in the last member is superfluous.
\end{prop}
The proof of \cref{lem:metricineq} can be found in \cref{sec:KMBinducedNorm}. From this proposition, and the minimum distance property of orthogonal projectors regarding the corresponding induced norm,
the following chain of inequalities holds:

\begin{eqnarray}\label{eq:chain_inequalities_norm}
\|\pi^\textnormal{KMB}({\bf Q})\|^\textnormal{KMB}&\leq&
\|\pi^{\rm covar}({\bf Q})\|^\textnormal{KMB}\\ 
&\leq&\|\pi^{\rm covar}({\bf Q})\|^{\rm covar}\\
&\leq& \|\pi^\textnormal{KMB}({\bf Q})\|^{\rm covar},
\end{eqnarray}

\noindent for any ${\bf Q} \in {\cal A}$. 
Equality holds when $B\subset \{|i\rangle\langle j|\}$, and hence the associated orthogonal projectors over ${\cal A}_B$ for $\pi\equiv \pi_B^\textnormal{KMB}$ and $\pi_B^{\rm covar}$ are identical.

From \cref{eq:chain_inequalities_norm} and since 
$${\pi}^{\rm KMB}{\pi}^{\rm covar}{\bf Q}={\pi}^{\rm covar}{\bf Q}\;\;\mbox{and}\;\; {\pi}^{\rm covar}{\pi}^{\rm KMB}{\bf Q}={\pi}^{\rm KMB}{\bf Q}\,,$$
if follows that
\begin{equation}
\delta^{\rm KMB}({\bf Q})\leq \delta^{\rm covar}({\bf Q})
\end{equation}
for $\delta^{\stackrel{\rm KMB}{\rm covar}}({\bf Q})=\|\pi^{\rm KMB}{\bf Q}-\pi^{\rm covar}{\bf Q}\|^{\stackrel{\rm KMB}{\rm covar}}$ and, using the triangular inequality,

\begin{equation}
    \begin{split}
        \delta^{{\rm KMB}}({\bf Q}) <  2 \|\pi_{\perp}^{{\rm covar}}{\bf Q}\|^{{\rm KMB}} \\
        \delta^{{\rm covar}}({\bf Q}) <  2 \|\pi_{\perp}^{\rm KMB}{\bf Q}\|^{{\rm covar}},
    \end{split}
\end{equation}

with $\pi_{\perp}^{^{\stackrel{\rm covar}{\rm KMB}}}{\bf Q}={\bf Q}-\pi^{^{\stackrel{\rm covar}{\rm KMB}}}{\bf Q}$ the corresponding projection onto the orthogonal complement of ${\cal A}_B$.

\subsection{Connection with standard formulations of Mean Field Theory and equivalence of projections in the Gaussian case}

As shown in \cref{subsec:proof_mf_as_proj}, for some special choices of $B$, our formalism is equivalent to the (self-consistent) Time-Dependent Mean Field Theory (TDMFT). \smallbreak
 
The simplest case is the one in which $\hilbert=\hilbertTensor\equiv \bigotimes_i \hilbert_i$ and the basis $B$ of accessible observables is a basis of local observables

\begin{equation}\label{eq:localbasis}
   B= \bigsqcup_i B_i,
\end{equation}

\noindent with $B_i$ complete bases of the local algebras of operators ${\cal A}_i$ acting over $\hilbert^{(i)}$. For this case, the formalism is equivalent to the Hartree (product-state based) mean-field approach ~\cite{Balian.1991,ULW.98,Bruus.2004,MRC.10}.

In a similar way, if $\hilbert = \hilbert^{\textnormal{Fock}}$, and

\begin{eqnarray*}\label{eq:gaussianbasis}
     B&=&B^{\textnormal{F}}_{1}\sqcup B^{\textnormal{F}}_{2}, \\
     B^{\textnormal{F}}_{1}&=&\{{\bf q}_1, {\bf p}_1, {\bf q}_2, {\bf p}_2,  \ldots\},     \\
     B^{\textnormal{F}}_{2}&=&\{{\bf Q}\textnormal{ }|\textnormal{ } {\bf Q}=[{\bf z}_i,{\bf z}_j]_{\mp}-\langle [{\bf z}_i,{\bf z}_j]_{\mp} \rangle, \textnormal{ } {\bf z}_{i,j}\in {B}^{\textnormal{F}}_{1} \},
\end{eqnarray*}

\noindent with ${\bf q}_i$, ${\bf p}_i$, observables s.t. $\langle {\bf p}_i\rangle=\langle {\bf q}_i\rangle=0$, ($[{\bf A},{\bf B}]_{+}=[{\bf A},{\bf B}]$ and $[{\bf A},{\bf B}]_{-}=\{{\bf A},{\bf B}\}$ correspond to the commutator and anticommutator of the operators)  satisfying canonical commutation/anti-commutation relations
  
   $$
        [{\bf p}_i,{\bf p}_j]_{\pm}=   [{\bf q}_i,{\bf q}_j]_{\pm}=0,\;  [{\bf q}_i,{\bf p}_j]_{\pm}={\bf i}\hbar \delta_{ij},
   $$
   
\noindent As a result, our formalism is equivalent to the Time-Dependent Hartree-Fock-Bogoliubov (Gaussian-state-based) Mean Field theory \cite{Auerbach.1994,Ring.2005}.
In both cases, the self-consistency condition -- for the stationary case -- is given by

$$
\Pi_{\mathcal{M}_B}(\sigma)=\sigma\;,\;\langle {\bf H}\rangle =
\langle \pi_B\left({\bf H}\right)\rangle.
$$

In other words, for the bosonic Gaussian case, both geometries yield exactly the same projection. \smallbreak

\paragraph{Possible simplifications using fixed referential mean-field states}

While beyond the scope of this article, there are further improvements that can be made to \cref{eq:coeffs}, besides altering the inner product. 
Specifically, instead of considering time-dependent scalar products w.r.t. the instantaneous state of the system, $\sigma(t)$, a single fixed and carefully chosen reference state $\sigma_0$ can be considered.

This proposal offers several advantages. 
For instance, by employing in \cref{eq:coeffs} a  \, \, \emph{covar} scalar product w.r.t. a fixed reference state $\sigma_0$, the resulting system of differential equations becomes linear. 
As a result, its solution becomes analytically tractable and numerically stable.

For this proposal to yield results comparable to the exact ones, the reference state $\sigma_0$ must exhibit a certain degree of similarity to the instantaneous states $\sigma(t)$ throughout the evolution. 
One way to achieve this is by considering a mean-field state as the reference state, i.e. $\sigma_0$ must be chosen s.t. 

$$
    \pi^{\textnormal{MF}}(\sigma_0) = \sigma_0,
$$

\noindent where $\pi^{\textnormal{MF}}: {\cal A} \rightarrow {\cal A}_B$ is the Mean-Field projector for the relevant basis of observables $B$. 
These ideas are discussed in depth in \cref{subsec:proof_mf_as_proj}. 
In the upcoming sections, these ideas will not be employed, and the scalar product will be computed w.r.t. the instantaneous state of the system.

\section{Test Example} \label{sec:test_ex}

By replacing the KMB scalar product with the correlation scalar product as depicted in \cref{eq:coeffsGramm}, one can derive expressions completely analogous to those presented  \cref{eq:coeffs} and \cref{eq:coeffsH}, albeit w.r.t. the aforementioned alternative scalar product.
Although the correlation scalar product exhibits mathematical similarity to the KMB scalar product and offers computational advantages, it remains to be seen whether it yields accurate results, when compared to both exact outcomes and those obtained through the KMB geometry. 
These ideas will be tested on a simple physical system, specifically the one-dimensional Heisenberg spin-$\frac{1}{2}$ chain, which will be summarized in the subsequent section. 
The objective is, then, to compare the exact results, obtained through numerical solutions of the Schr\"odinger equation \cref{eq:Schrodinger}, with those derived from the KMB geometry and the geometry induced by the correlation scalar product.

\subsection{XX Heisenberg Model}\label{section_heisenberg}

\smallbreak 

As an illustrative instance of the preceding formalism, let us contemplate a spin-$\frac{1}{2}$ nearest-neighbour Heisenberg XX model on a periodic chain one-dimensional lattice composed of $N$ sites.
The system is governed by a Hamiltonian, given by:

\begin{equation}
    {\bf H} = -J \bigg(\sum_{j=1}^{N} \spin_j^x \spin_{j+1}^x + \spin_j^y \spin_{j+1}^y\bigg). 
    \label{XX_hamiltonian}
\end{equation}
s.t.\ $\spin_{N+1}^{x,y,z} \equiv \spin_1^{x,y,z}$ and 
where $\{\spin_j^x, \spin_j^y, \spin_j^z\}$ are the usual spin-$\frac{1}{2}$ operators and where $J$ is the strength of the flip-flop term $\spin_{j}^x\spin_{j+1}^x + \spin_j^y \spin_{j+1}^y$. 
Note that the $\spin_n^{x,y,z}$ operators act non-trivially just on the $n$-th site.
This state of the system can be described using (linear combinations of) tensor products of $N$ 
$\mathfrak{s}\mathfrak{u}(2)$ 
representations, with the identity operator, added for each lattice site. 
Its Hilbert space is $2^N$-dimensional, where one possible configuration is $\ket{\uparrow_1 \uparrow_2 \cdots \downarrow_N}$. 
In a quantum information context, these states are known as the  computational basis vectors. 
Moreover, both the XX and the more general, XY model can be analytically diagonalized via a Jordan-Wigner transformation~\cite{JW.28,LSM.61}. 
However, computing time-dependent numerical correlations, which are important for understanding these model's low-temperature behavior -amongst other important physical features-~\cite{HHSV.17,Lee_2007}, requires a numerical computation, wherein the previous technological difficulties readily become apparent. \smallbreak{}

\paragraph{Observables and quantum numbers} \smallbreak{}

Since the total magnetization ${\bf S}^z_T = \sum_{i} \spin_i^z$ commutes with the Hamiltonian, all states may be labelled with an additional quantum number, indicating the total number of excitations present in a given configuration, relative to the reference state~\cite{Auerbach.1994}
$$|\downarrow \downarrow \ldots \rangle\equiv |0\rangle.$$ 
Furthermore, the magnetization is a conserved quantity and, hence, the  Schr\"odinger evolution preserves it, i.e. a state with $n$ excitations will evolve in time to states with exactly $n$ excitations, as well. 

Consider, then, the following operator, basically a redefinition of the global magnetization, 
\begin{equation}
    {\bf n} = \sum_{j=1}^{N} \bigg(
        {\bf S}^z_j + \frac{1}{2}
    \bigg).
\end{equation}
This operator, the \emph{occupation operator}, measures how many flipped excitations the system contains, w.r.t to the reference state $|\downarrow \downarrow \ldots \rangle$, and is a constant of motion. 
In particular, consider a system with initial state $\rho_0$ s.t.\ $\langle \hat{\bf n} \rangle_{\rho_0} = 1$, undergoing a Schr\"odinger evolution.
Then, $\langle \hat{\bf n} \rangle_{\rho(t)} = 1$ at all times. 

A second quantum number of interest  is the average (normalized) localization of the excitations, given by a position operator
\begin{align}
     &{\bf x} = \sum_{j = 1}^{N} \bigg(2\frac{j-1}{N-1}-1\bigg) \bigg(\spin^z_j + \frac{1}{2}\bigg),
\end{align}

\noindent which measures which lattice site contains the excitation. 
This accessible observable will be of relevance in the following section. \smallbreak

\subsection{Numerical Exploration of the Projected Dynamics} 

Thus far, two potential alternatives for dynamics involving projections have been introduced, the projected and restricted evolutions. 
The former are derived by projecting the exact (free) dynamics \cref{eq:SchrodingerK} onto the Max-Ent manifold, as described in equation \cref{eq:piasbesselfourier}.
On the other hand, the restricted evolutions are obtained by solving the restricted equation of motion, as stated in \cref{eq:coeffs}, utilizing different types of linear projectors, $\pi^s_B$ with $s = \textnormal{KMB}/{\rm covar}$.
This also serves to gauge how well justified the hypothesis of substituting the KMB geometry by the  \, \, \emph{covar} geometry is. 
For the comparisons, we are going to consider the ferromagnetic case $J=1$ of a six-site chain with periodic boundary conditions (here, physical quantities are given in natural units, so $\hbar=1$).
The corresponding Hilbert space $\hilbertTensor = \otimes_{j} \hilbert^{(j)}$ is 64-dimensional, high-dimensional enough for exact numerical methods to be applicable but cumbersome and computationally expensive as well.
The free dynamics was obtained by numerically solving \cref{eq:SchrodingerK}, using the Quantum Toolbox in Python's (QuTip) function master equation solver~\cite{johansson_qutip_2013}. 
Restricted dynamics were computed using the explicit Runge-Kutta 5th-order solver, from the Scipy library.

\paragraph{Max-Ent manifold.}
In the examples, it was considered a basis of relevant observables including the constants of motion ${\bf n}$, ${\bf n}^2$, and ${\bf H}$, as well as the position operator ${\bf x}$ and its square ${\bf x}^2$.
From \cref{prop:trivial_projection_in_expect_values}, ${\bf id}_{\hilbert}$ must be included to ensure that, in the asymptotic limit, the action of the projection does not modify the expectation value of any operator. On the other hand, the constants of motion ${\bf n}$, ${\bf n}^2$, and ${\bf H}$ are included both because we want to study its behavior in the projected dynamics, and because from \cref{prop:ehrenfest}, its inclusion ensures its conservation also in the KMB restricted dynamics. Finally, the pseudo-position operator ${\bf x}$ and its square are included as an example of relevant quantity that is not conserved in the free dynamics.

This set is enlarged by including the iterated commutators ${\rm ad}^{(\ell)}_{{\bf H}}({\bf K})$ up to $\ell=4$ (see \cref{sec:apphierarch}), in a way that $B=B^{it}_4$ with 

\begin{equation}\label{eq:iteraterB}
    \begin{split}
        &B_\ell = B_0 + B^{it}_\ell, \\
        &B_0 = \{\textnormal{id}_{\hilbert}, {\bf n}, {\bf n}^2, {\bf x}, {\bf x}^2,{\bf K}_0\}, \\
        &B^{it}_\ell = \{\underbrace{[{\bf H}_0, {\bf K}_0]/({\bf i}), \big[{\bf H}_0, [{\bf H}_0, {\bf K}_0]\big]/({\bf i})^2, \cdots}_{\textnormal{a total of $\ell$ times}} \}.
    \end{split}
\end{equation}

\smallbreak

\paragraph{Initial Conditions.}The initial state of the system $\rho_0 = \rho(0)$ is chosen to lie in the Max-Ent manifold and is given by

\begin{equation}\label{eq:K0test}
\begin{split}
    \rho(0) &\propto e^{-{\bf K}(0)}, \\
    &{\bf K}(0) = \beta {\bf H} + c_1(\hat{\bf n} - \zeta)^2 + c_2 (\hat{\bf x} - {\bf x}_0)^2,
\end{split}
\end{equation}

\noindent where $\beta$ is the inverse temperature.
Two values of $\beta$ are of interest: $\beta = J$ and $\beta = J/10$. 
Here, $\rho(0)$ is a Max-Ent state regarding the observables $B = \{{\bf H}, {\bf n},{\bf n}^2, {\bf x}, {\bf x}^2\}$.
The other coefficients are chosen s.t. $\langle \hat{\bf n}\rangle_{\rho_{0}} \approx 1$.
In the first case, $c_1 = 3\beta, c_2 = 3 \beta, \zeta = 1$ and $x_0 = -.3$ have been chosen, 
while in the second case, $c_1 = 10\beta, c_2 = 10\beta, \zeta = 1$ and $x_0 = -.3$.

\smallbreak 

\subsubsection{Projected Evolutions}

Having defined the test case, the first question is whether or not linearized projections provide a sensible approximation to the Max-Ent projection. 
From the analysis in \cref{subsec:geometrycov}, this is assured if the projections $\sigma=\pi(\rho(t))$ are close enough to the original state $\rho(t)$. 
For our choice of basis $B$, this is asymptotically true for the short-time evolutions. 
This is also assured when the restricted evolution is close to the free dynamics.

\paragraph{Geometric distance between different projections.}

As a first step, we are interested in quantifying the loss of accuracy in the results when switching from the KMB to the  \, \, \emph{covar} geometry, for different temperatures. 
These results are depicted in \cref{fig:KMBKnorms}.
It is evident from the data that, for short-term evolutions, all three evolutions exhibit minimal, albeit non-zero, differences.
Given our choice of basis, $B = B_\ell$ in \cref{eq:iteraterB}, which includes up-to the $\ell$-th iterated commutator of ${\bf K}(0)$ and the Hamiltonian ${\bf H}$, $\pi_B$ acts trivially over the ${\bf K}$-power expansion up-to ${\cal O}(t^\ell)$.
As a result, the projected dynamics (and the expectation values derived from it) deviate from the exact free evolution in amounts of the same order. 
On the other hand, as the evolutions extend to longer durations, the discrepancy between the projected states and the exact states increases, eventually reaching a saturation point at around the $tJ\simeq 10$ mark. 
In contrast, the geometric distance between the  \, \, \emph{covar}- and KMB- ${\bf K}$-states remains minuscule, in comparison, during the entirety of the simulation.
In general, one notes that the projected states remain in close proximity to the exact states, albeit at a growing distance.
These observations support our proposal of substituting the computationally expensive KMB geometry with the  \, \, \emph{covar} geometry, at least for short-term evolutions.

\begin{figure}[htp]
    \centering
    \includegraphics[width=\linewidth]{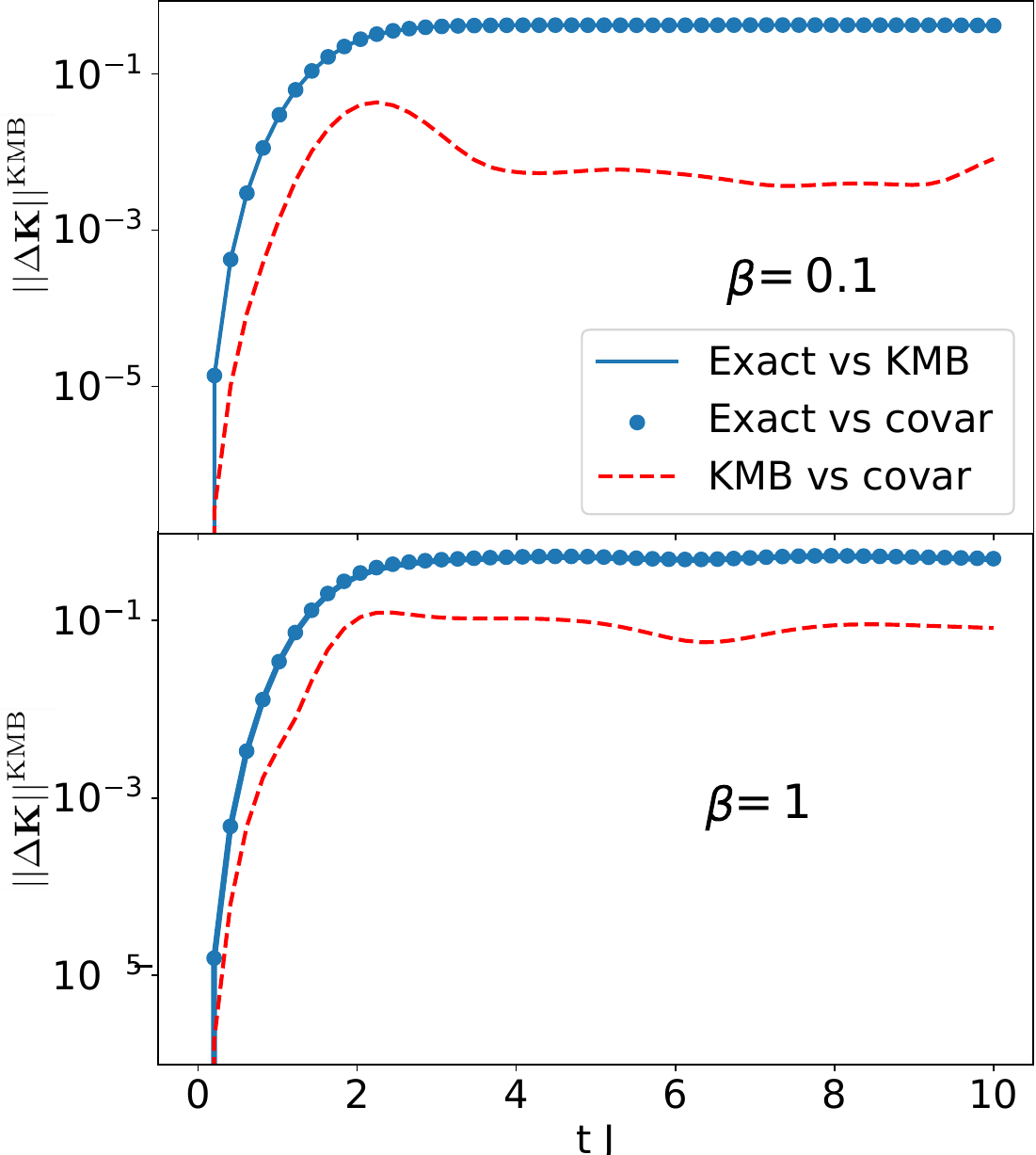}
    \caption{Evolution of the KMB-induced norms between the exact logarithm of the states and their KMB/correlation projections, in logarithmic scale, obtained from a $t=10/J$ simulation with 200 steps. We show these results for two inverse temperatures, $\beta = J/10$ (top) and $\beta = J$ (bottom). 
For the short-term evolution, both the KMB and  \, \, \emph{covar} projected states exhibit remarkable similarities amongst themselves and with the exact state.}
\label{fig:KMBKnorms}
\end{figure}

\smallbreak

\paragraph{Time evolution of Expectation Values.} 

As established by \cref{prop:ortproj}, KMB distances provide bounds to the deviations in the estimation of \emph{any} possible observable regarding the original and the projected state.
In  \cref{fig:obs_n_ev,fig:obs_H_ev,fig:obs_x_ev,fig:obs_Hk5_ev} the time evolution of the expectation values associated with some representative observables, regarding the different projections, are depicted. 
These plots correspond to a simulation of duration $t = 10/J$, employing a grid of 200 points.

By construction, the Max-Ent (non-linear) projection \cref{eq:nlProj} preserves the expectation value of any observable in the relevant space ${\cal A}_B$. 
On the other hand, linear projections $\pi^{\rm KMB}$ and $\pi^{\rm(covar)}$ satisfy \cref{eq:nlProj} only in the neighborhood of ${\cal M}_B$. 
 
\begin{figure}[!ht]
    \centering
    \includegraphics[width=\linewidth]{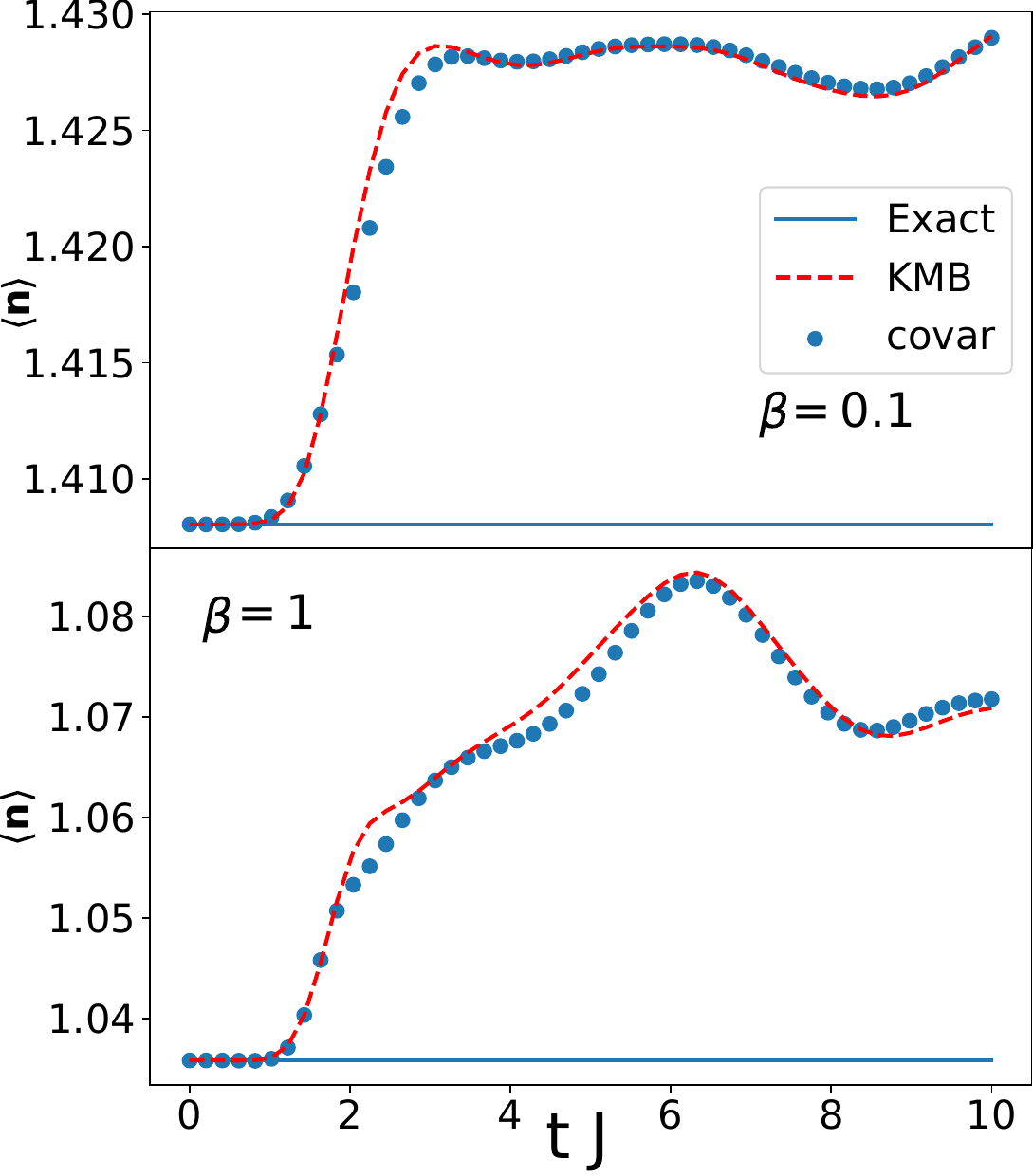}
    \caption{
Evolution of the expectation values for the occupation operator ${\bf n}$ at inverse temperatures $\beta = 0.1$ (top) and $\beta = 1$ (bottom), regarding the exact state and their linearized Max-Ent projections (see \cref{eq:piasbesselfourier}) concerning the basis $B=B_4$ given by \cref{eq:iteraterB}.
In the short-term regime ($t J \lessapprox 2$), the three dynamics yield highly similar outcomes. 
The subsequent lack of conservation is a consequence of the departure of the exact state trajectory from the corresponding Max-Ent manifold ${\cal M}_B$ and the limitations of the linear approximation.
}
\label{fig:obs_n_ev}
\end{figure}

\begin{figure}[htp]
    \centering
    \includegraphics[width=\linewidth]{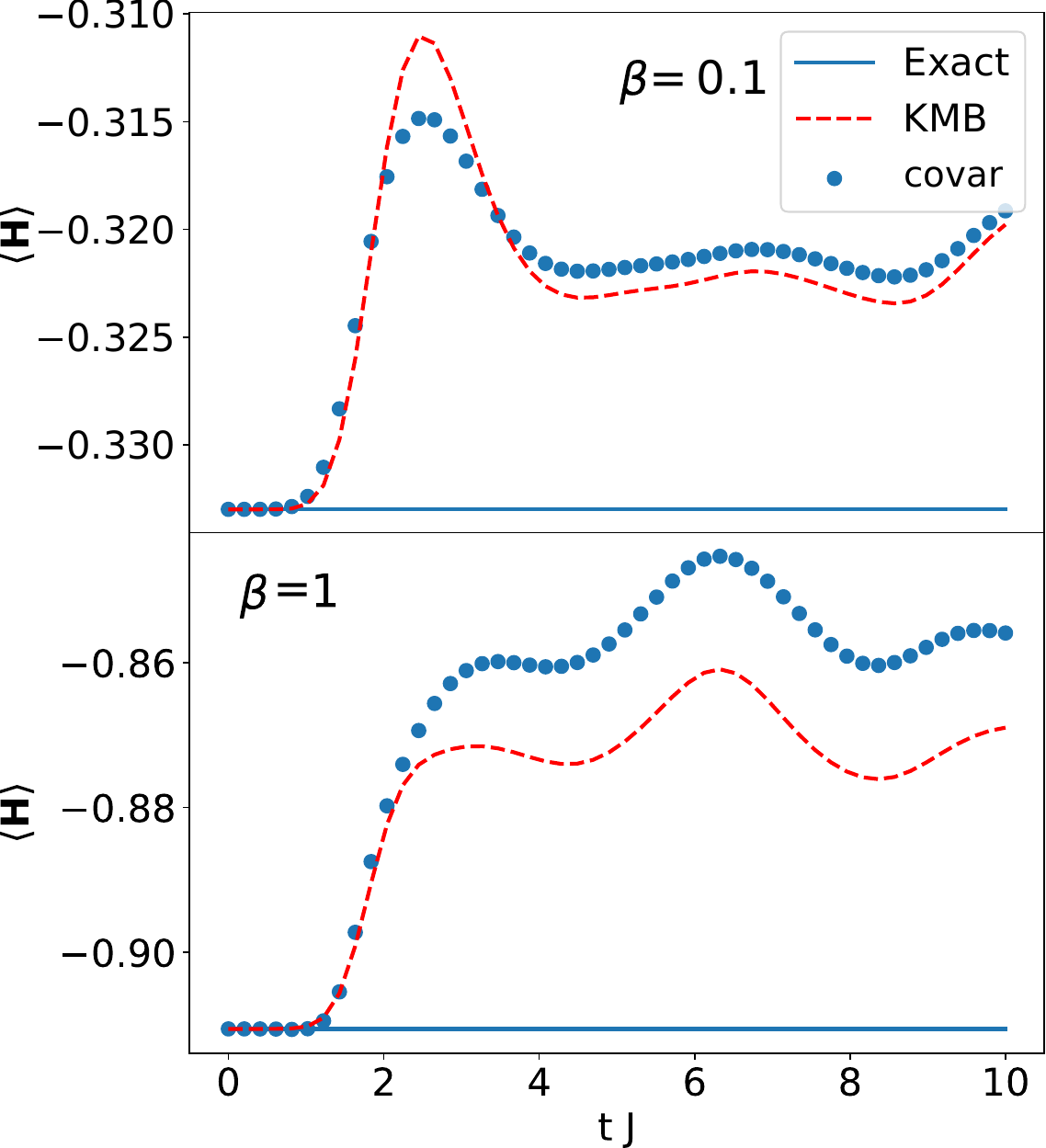}
    \caption{
Evolution of the expectation values for the Hamiltonian operator ${\bf H}$ at inverse temperatures $\beta = 0.1$ (top) and $\beta = 1$ (bottom), regarding the exact state and their linearized Max-Ent projections (see \cref{eq:piasbesselfourier}) concerning the basis $B=B_4$ given by \cref{eq:iteraterB}.
In the short-term regime ($t J \lessapprox 2$), the three dynamics yield highly similar outcomes. 
The subsequent lack of conservation is a consequence of the departure of the exact state trajectory from the corresponding Max-Ent manifold ${\cal M}_B$ and the limitations of the linear approximation.}
\label{fig:obs_H_ev}
\end{figure}

\begin{figure}[htp]
    \centering
    \includegraphics[width=\linewidth]{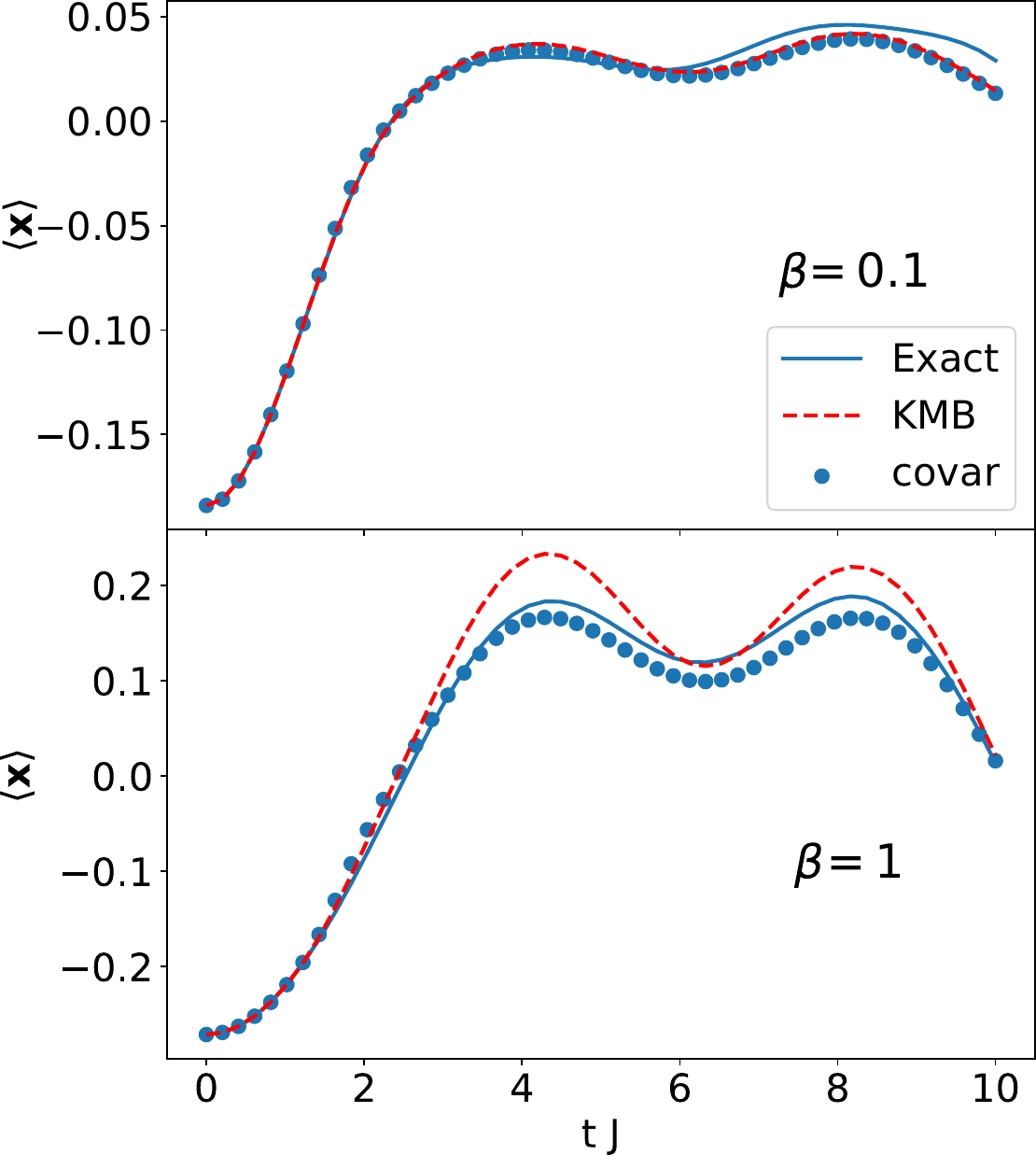}
    \caption{
Evolution of the expectation values for the operator ${\bf x}$ at inverse temperatures $\beta = 0.1$ (top) and $\beta = 1$ (bottom), regarding the exact state and their linearized Max-Ent projections (see \cref{eq:piasbesselfourier}) concerning the basis $B=B_4$ given by \cref{eq:iteraterB}.
In the short-term regime ($t J \lessapprox 2$), the three dynamics yield highly similar outcomes.
The subsequent lack of conservation is a consequence of the departure of the exact state trajectory from the corresponding Max-Ent manifold ${\cal M}_B$.}
\label{fig:obs_x_ev}
\end{figure}

\begin{figure}[htp]
    \centering
    \includegraphics[width=\linewidth]{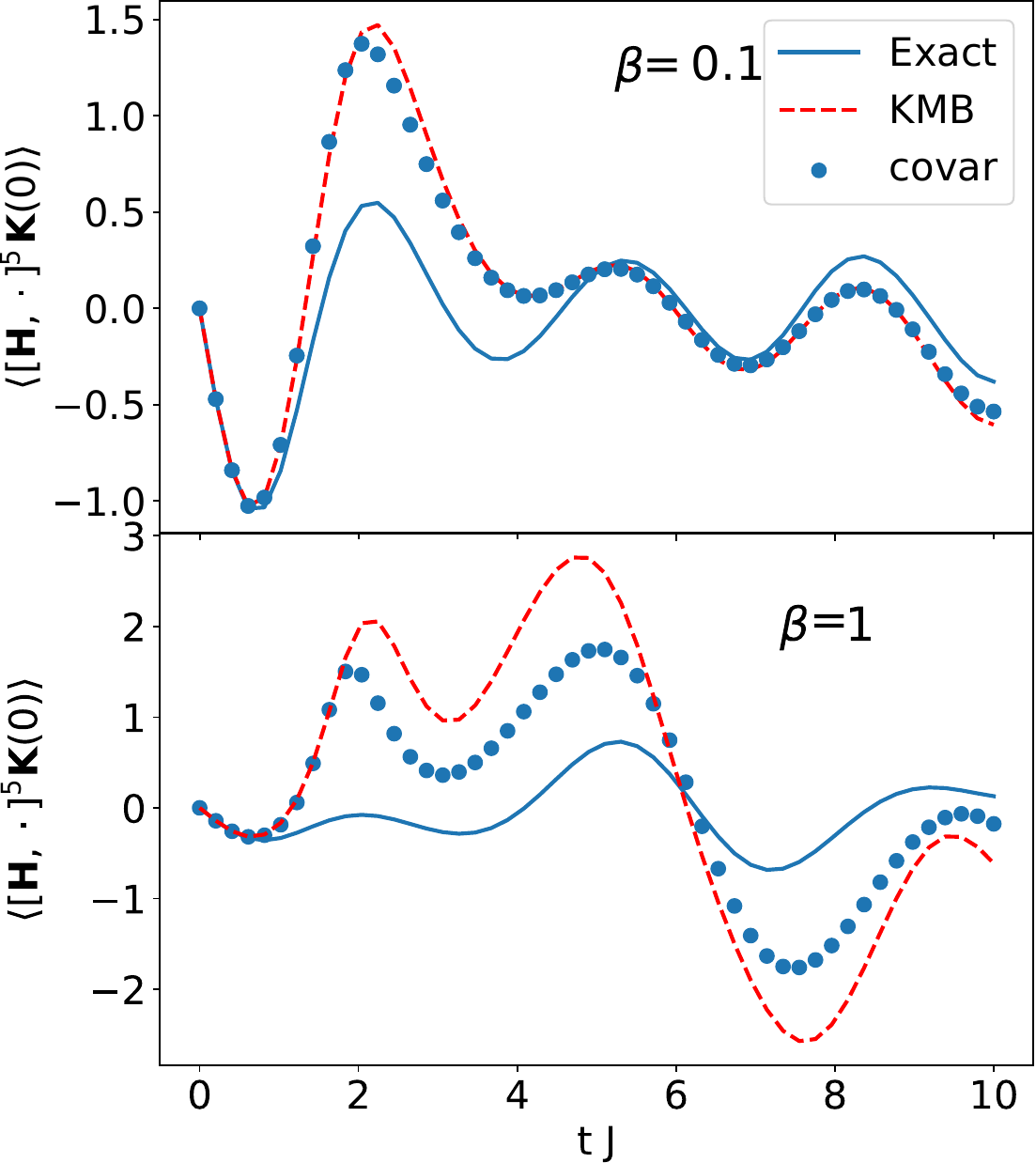}
    \caption{
Evolution of the expectation values for the $5-$times iterated commutator of the Hamiltonian with ${\bf K}(0)$ at inverse temperatures $\beta = 0.1$ (top) and $\beta = 1$ (bottom), regarding the exact state and their linearized Max-Ent projections (\cref{eq:piasbesselfourier}) concerning the basis $B=B_4$\cref{eq:iteraterB}.
In the short-term regime ($t J \lessapprox 2$), the three dynamics yield highly similar outcomes. 
The subsequent lack of conservation is a consequence of the departure of the exact state trajectory from the corresponding Max-Ent manifold ${\cal M}_B$.}
\label{fig:obs_Hk5_ev}
\end{figure}

Before the $t J = 2$  mark, all three frameworks exhibit highly similar constant outcomes, indicating a strong agreement between the projected and exact frameworks. 
However, as the simulation progresses, discrepancies between the projected and exact frameworks become more pronounced and eventually reach a saturation point around the $t=10/J$ mark. 
Notice that the no-conservation of ${\bf n}$ and ${\bf H}$ is an effect of the linearization of the Max-Ent projection ${\cal P}_{B}$, which is more important as the state moves away the Max-Ent manifold ${\cal M}_B$. 
Nevertheless, the difference between the KMB- and \emph{covar}- projections remain small throughout the evolution, meaning that both projections produce comparable results. Moreover, for the larger deviations, it can be noticed that \, \emph{covar} projection sometimes produces better results than the KMB projection.
This is a consequence of the competition of the error introduced by the linearization, and the one resulting from approximating the KMB projection by the \, \emph{covar} projection.

Additionally, notice that we have not established any asymptotic behavior for the \, \emph{covar} scalar product.
This is not problematic since we are examining \textit{large} regions of ${\cal A}$ in short-term evolutions.
In the three cases, it is observed that the deviations from the exact values are larger at lower temperatures (larger $\beta$). This is an expected behavior since the error bound \cref{eq:error_bound} is proportional to ${\bf K}$ and therefore, to the inverse temperature $\beta$.

For the case of the occupation number ${\bf n}$ (see \cref{fig:obs_n_ev}),  which commutes with both ${\bf H}$ and ${\bf K}(t)$, deviations can only be attributed to the effect of neglecting the nonlinear terms in the projectors. 
For $\beta=0.1$ the non-conservation is below $1.5\%$, while for $\beta=1$ it is under $4\%$ of the initial value. 
Notice that these fluctuations are also affected by the non-conservation of the normalization $\Tr \mathbf{id}_\hilbert\exp(-\pi_B {\bf K}(t))\neq \Pi_B \exp(-{\bf K}(t))$ which suffers the same effect.
We also notice that the deviations obtained from both the KMB and the  \, \, \emph{covar} linear projections lead to very similar values.

For the case of the Hamiltonian ${\bf H}$ (see \cref{fig:obs_H_ev}), which does not commute with ${\bf K}(t)$, the behavior is similar, but differences between the values obtained with the two projectors become larger, especially at the lower temperature.
Deviations regarding the initial value are below $7\%$ for both temperatures.

The case of the (pseudo) position operator (${\bf x}$), which commutes with neither ${\bf H}$ nor ${\bf K}(t)$ (\cref{fig:obs_x_ev}) presents a similar behavior, with an excellent agreement in the short term regime, but with larger deviations for longer times, which for $\beta=1$ becomes close to the $30\%$.
Interestingly, the \, \, \emph{covar} projection provides in this case closer values to the exact ones than the KMB.
Again, this seems to be the result of error cancellations happening beyond the linear regime.

Finally, in \cref{fig:obs_Hk5_ev}, the expectation value of the $5-$times iterated commutator
${\bf b}_5={\rm ad}_{{\bf H}/i}^{(5)} {\bf K}_0$ is depicted for the free state, its (non-linear) Max-Ent projection, and the linear projections.
In this case, the operator does not belong to ${\bf A}_B$, and hence, the expectation value for the free state and its Max-Ent projection do not necessarily match.
Again, as predicted, all the averages coincide in the short time regime ($t \lessapprox 1/J$) but start showing deviations at shorter times. 
For larger times, the expectation values corresponding to different projection schemes show larger fluctuations than those corresponding to the free state. 
Interestingly, deviations from the free dynamics result larger for the KMB projection, and even for the true Max-Ent projection, than those obtained from the \, \emph{covar} projection.

This underscores the effectiveness of linearization as a reliable quantitative approximation.
\smallbreak{}

\paragraph{Relative Entropies}

Finally, we are interested in quantifying relative entropies between the exact free evolution and both kind of projections. 
In particular, the following relative entropies are of interest: 

\begin{enumerate}
    \item the relative entropy between the exact and the KMB-projected states, $S(\rho || \sigma_{\textnormal{KMB}})$,
    \item the relative entropy between the exact and the correlation-projected states, $S(\rho || \sigma_{{\rm covar}})$,
    \item and both types of relative entropies between the correlation- and KMB-projected states, $S(\sigma_{{\rm covar}} || \sigma_{\textnormal{KMB}})$ and  $S(\sigma_{\textnormal{KMB}} || \sigma_{{\rm covar}})$.
\end{enumerate}

\begin{figure}[htp]
    \centering
    \includegraphics[width=\linewidth]{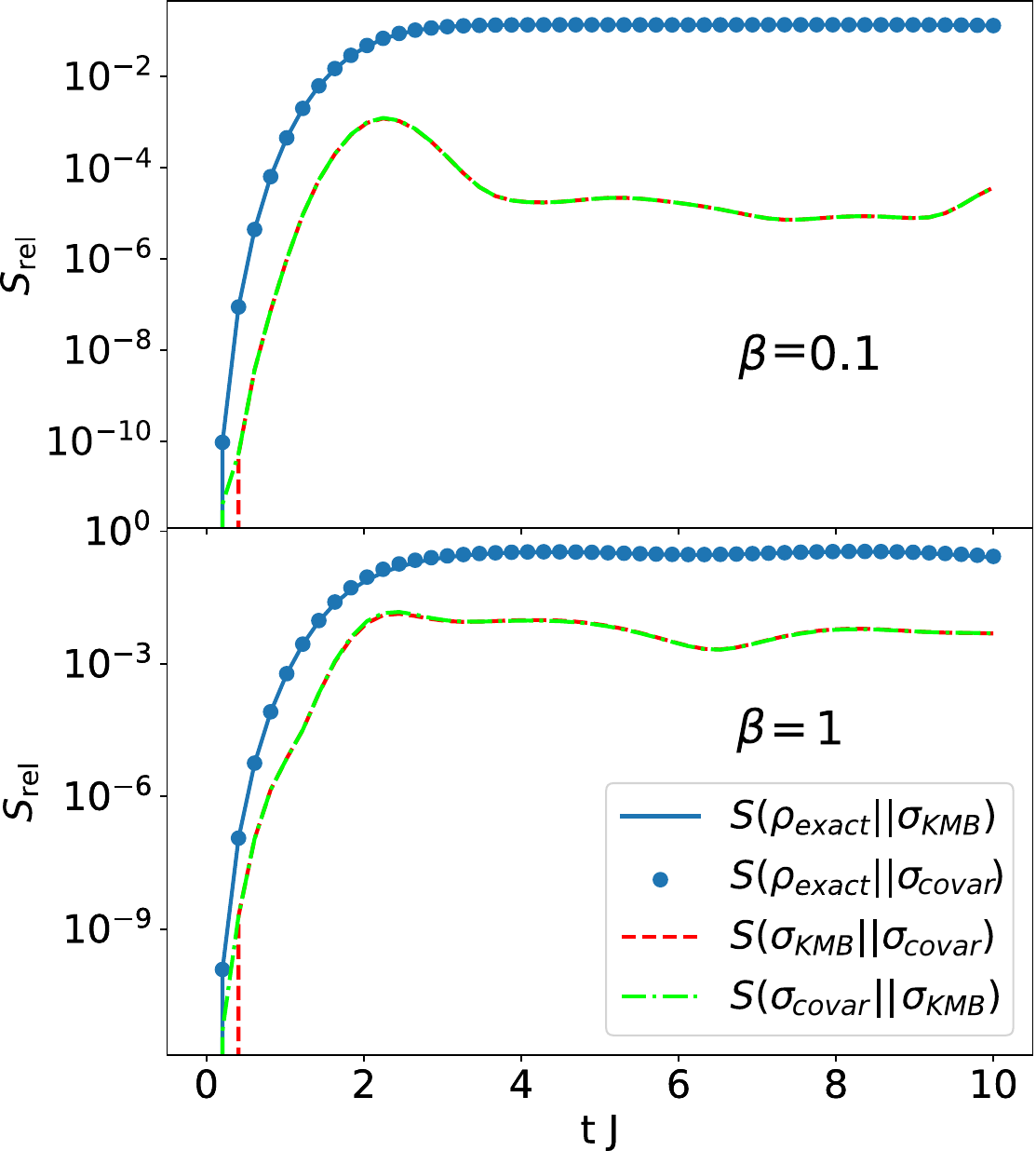}
    \caption{\label{fig:relent}
    Relative entropies between exact and projected states, and between the projected states at inverse temperatures $\beta = 0.1$ (top) and $\beta = 1$ (bottom).
Once again, for short-term evolutions, the three frameworks yield very similar results. 
For longer-term evolutions, the relative entropies saturate. }
\end{figure}

In \cref{fig:relent}, the evolution of the relative entropies are depicted for two different temperatures.
Consistent with the findings presented in \cref{fig:KMBKnorms}, it can be observed that for short-term evolutions, the relative entropies between different states exhibit very small values, indicating a high level of similarity among these states. However, for larger times, the relative entropies between the exact and projected states become more noticeable, eventually reaching a saturation point. This behavior aligns closely with the trends observed in the geometric distances between the three classes of states, as shown in \cref{fig:KMBKnorms}. Furthermore, it is worth noting that the relative entropies between the KMB-projected and correlation-projected states remain consistently negligible throughout the entire evolution, further underscoring the strong agreement between these two frameworks.\smallbreak

\subsection{Projected vs restricted Dynamics}

So far, the comparison has been focused on the exact (free) dynamic and its KMB and \, \emph{covar} projections over the Max-Ent manifold. 
Let's compare them, now, against the solutions to the restricted equation of motion (see \cref{eq:coeffs}) obtained from the KMB and \, \emph{covar} instantaneous projections, computed using the orthogonal expansion of \cref{eq:piasbesselfourier}.

\cref{fig:exactvsconstrained} illustrates the KMB-induced norm between the state of the system (top) and the expectation value of the ${\bf x}$ operator (bottom) for both the exact and projected dynamics and the KMB/ \, \emph{covar} restricted dynamics.

\begin{figure}
\includegraphics[clip,width=\linewidth]{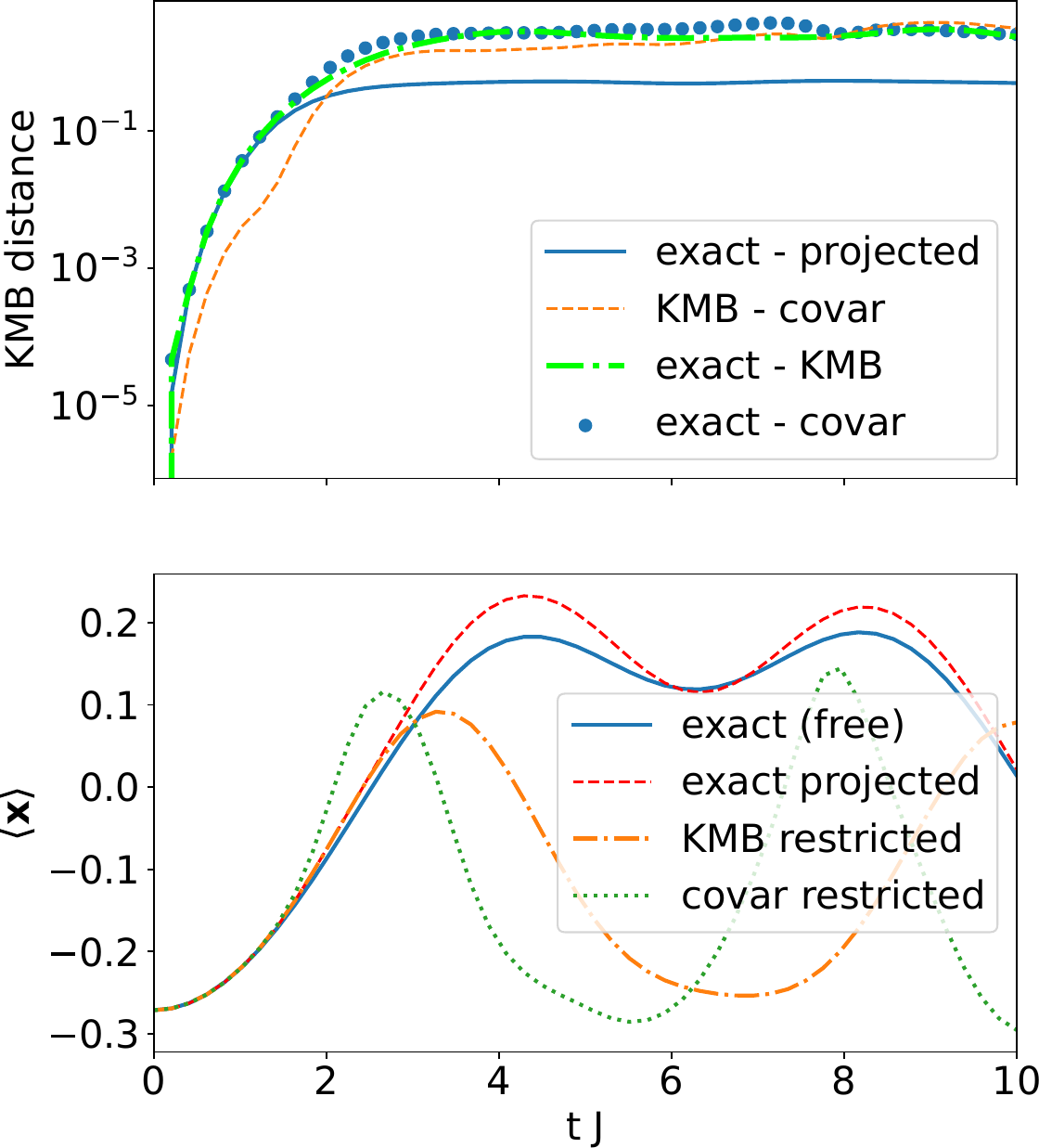}
\caption{Comparison of the KMB-induced distance between the state, the projected state, and the state arising from the restricted evolution according to the KMB and  \, \emph{covar} geometries (top). 
Evolution of the expectation value of the ${\bf x}$ operator for the corresponding states for $\beta=1$ (bottom). 
For short-term evolutions, lasting less than $tJ=2$, all four evolutions yield highly similar results. Subsequently, noticeable distinctions become apparent in the outcomes.
}\label{fig:exactvsconstrained}
\end{figure}

As is expected, for short times, the exact free dynamic is asymptotically close to both the KMB projected and the restricted dynamics, disregarding the choice of the projectors. For longer times, the behavior of the restricted dynamics is similar among them, being sometimes the covariance-restricted evolution is slightly closer to the exact dynamics than the KMB. In the figure, it happens at $tJ\approx 8$. In all the cases, as expected from \cref{eq:linprojMinCondition}, the KMB linear projection is always closer to the exact than any one of the restricted evolutions. On the other hand, the instantaneous states obtained from the restricted evolutions are typically closed between them than to the exact state. 

Similar conclusions can also be drawn concerning the expectation value of the position operator: In this case, the projected dynamic reproduces its behavior more closely than the restricted evolutions, both in phase and amplitude. Width and distances between peaks and valleys in the expectation value also present differences among the exact and the different restricted evolutions but keep a qualitative agreement.

Notice however that the computational cost of solving \eqref{eq:restricted} with the KMB projector is much larger than the required effort to solve it using the  \, \emph{covar} projector. 
For example, to compute the values depicted in \cref{fig:exactvsconstrained}, solving the KMB restricted dynamics involved 36hs of computations, against the 10 minutes required to solve the same equation for the \, \emph{covar} projection \footnote{All the computations implemented using Qutip 4.7.1 and ODE solver of Scipy 1.10.1 over Python 3.11, on a Intel i5-12400 core.}.

While not addressed in depth in this article, it is worth mentioning that both the  \, \emph{covar}-restricted and KMB-restricted evolutions could be enhanced in terms of computational efficiency by approximating the instantaneous reference state with product states by means of a mean-field approximation of the instantaneous state. Additionally, the \, \emph{covar} enables the introduction of additional correlations using separable states as reference states.
These approximations could allow us to solve much larger problems without significant overhead. However, the aim here is to compare the effect of the different choices of projections, which could be masked by these further simplifications.

On the other hand, to understand and quantify the origin of the deviation between the exact free evolution and the projected and restricted evolutions, it is worth inspecting the effect of the projections over the expansion of ${\bf K}$ as a power series on $t$. As we have mentioned before, by including the basis elements $B_{L}$, see \cref{eq:iteraterB}, the free, projected, and restricted evolutions coincide up to ${\cal O}(t^L)$. The restricted evolution case can be seen as an approximation of the exact evolution that consists of inferring the $L+1$ order time-derivative of the observables defining the state in terms of the lower-order derivatives. 
\cref{fig:obs_Hk5_ev} depicts how the estimation of this derivative obtained by projection departs from the obtained from the free evolution. The results of the simulations suggest that the estimation resulting from the KMB and \, \emph{covar} linearized projections are similar in accuracy, which allows us to choose just the more computationally convenient.  
\smallbreak

\section{Discussion}\label{sec:discussion}

In this article, we discuss a framework for constructing stable and efficient approximations for Max-Ent states and the associated dynamics.
To achieve this, we reformulate the Max-Ent optimization problem in terms of nonlinear projecting functions and their linear approximations, in the neighborhood of the corresponding Max-Ent manifold.
The requirement of local exactness in the linearized projection leads to identifying it with the orthogonal projection associated with the KMB metric space. 
The substantial computational cost in the evaluation of this projection motivates us to reexamine the relevant features of the KMB geometry and propose a different linearization scheme, which was the main subject of this work. 

Following this line, the  \, \emph{covar} geometry (\cref{subsec:geometrycov} ) was introduced. 
This scalar product shares fundamental characteristics with the KMB product, while offering analytical properties that allow a more stable and efficient evaluation, and for a larger class of reference states, such as the separable states. 
This stems from the fact that its evaluation does not rely on the explicit evaluation of the spectral decomposition of the reference state.

Thereafter, analytical relations and bounds between the induced geometries of both products were studied.
In particular, the equivalence of the induced orthogonal projections, for Max-Ent manifolds of product and Gaussian states, was shown. 
It has also been proven that the self-consistent and time-dependent mean-field approximation can be equivalently expressed in terms of orthogonal projectors, defined with either of the two scalar products. 

Based on these results, the application of this formalism to the study of the dynamics of closed quantum systems projected on Max-Ent manifolds was considered, and its approximation by the restricted dynamics on said manifolds. Also, convergency criteria among the exact, the Max-Ent projected and the proposed restricted evolutions were discussed.  This led to expressions analog to those discussed in \cite{BAR.86}, but regarding the \, \emph{covar} product. 

As an application, the dynamics of excitations over a spin chain evolving with a Heisenberg XX Hamiltonian were analyzed.

Similarities and differences between the evolved state and the result of applying various linear projectors over it were investigated. 
As expected, for the considered cases, it was observed that, according to the relative entropy, for short times, both schemes of projection are indistinguishable from the free evolving state, while at longer times, the deviation between the projected and exact dynamics reached saturation values, significantly larger than the discrepancies between the two projections. 
The estimation of expectation values presents a consistent behavior with the distinguishability measures.

Another noteworthy observation in these plots is that the expectation values obtained under the KMB projection, are not consistently closer to those obtained from the original state than the values obtained from the \, \emph{covar} projection. 
This observation may initially appear counterintuitive, since as the KMB-based projection is locally exact, it would be the best approximation to the true projector. 
However, the regions where this inversion in the expected order happens coincide with the parts of the trajectory in which both approximations are furthest from the exact value. 
This suggests that for states not-too-close to the Max-Ent manifold, i.e, the region where non-perturbative effects happen,
both linear projectors provide similarly robust approximations. 

Subsequently, the free dynamics, its KMB projection, and the restricted dynamics regarding KMB and \, \emph{covar} projections were compared.
Once again, it was observed that the restricted dynamics are qualitatively similar and converge asymptotically for short times.
Nevertheless, the dynamics diverge at longer times, until a saturation distance is reached. 
Additionally, it is noted that although the plots of the expectation values exhibit similar behaviors, the observed oscillations in the graphs are noticeably different. 
This observation aligns with the discussion about the restricted dynamics for an individual spin.
Furthermore, it was found that the restricted dynamics obtained from both geometries are remarkably more similar to each other than either of them in the exact projection. 
This supports the hypothesis that the \emph{covar}-restricted dynamics is a good approximation to the KMB dynamics.

These findings suggest that, while the KMB orthogonal projection represents the actual, consistent local linearization of the Max-Ent projection, the {\color{magenta} \, \emph{covar}} orthogonal projection can yield similar results with less computational cost. 
Moreover, when the exact dynamics are not strictly confined in the close neighborhood of the Max-Ent manifold, both approximations to the Max-Ent projection are similarly good.

To make this approach suitable for efficiently computing quantum simulations of larger quantum systems, the next step would involve replacing the reference states with more efficient approximations of the instantaneous, correlated states. 
Exploring such alternatives is the focus of our upcoming work, which is currently in preparation.

Another aspect that deserves future work is how this formalism can be applied to open quantum systems. For the sake of clarity, in this work, all the discussion was constrained to the case of closed quantum systems, which do not present non-trivial fixed points. This avoided making stronger statements about the convergence in the long-time dynamics. In return, it helped to highlight aspects related to the effect of the choice of the geometry and the linearized projections in the Max-Ent approximate dynamics, the main topic of the article. The relevant case of open dynamics is going to be addressed in forthcoming research.

\section{Acknowledgments}

The authors would like to thank Joaquín Pelle for interesting discussions on the topology of quantum manifolds, Matías Bilkis for useful remarks about this work, Tomás Crosta for interesting discussions concerning the numerical aspects of our work, as well.
The authors acknowledge support from CONICET of Argentina.
Work supported by CONICET PIP Grant No. 11220200101877CO.

\begin{appendix} 

\section{Proof of \cref{prop:Max-Ent_are_gibbs}}\label{subsec:mfgibbs}

To provide a self-contained presentation, in this section, a proof of \cref{prop:Max-Ent_are_gibbs}, which can be found in other references (see for instance \cite{Jaynes1, Balian.1991}), is reproduced.

\begin{proof}
Given that $\sigma^\star \in {\cal B}(\hilbert)$, an open set, and given that both the target function and the constraints are continuous, differentiable functions, $\sigma^{\star}$ must satisfy the stationary condition, 

$$
    \delta\left. \bigg(S(\sigma)-\sum_{\alpha=0}^N \lambda_{\alpha} C_{\alpha}(\sigma)\bigg)\right|_{\sigma=\sigma^\star}=0,
$$

\noindent where  $\lambda_{\alpha}$ are Lagrange multipliers reinforcing the conditions  $C_{\alpha}(\sigma)=\Tr\sigma {\bf Q}_{\alpha}-\mu_\alpha=0$ that define $\compatstates{B}{\mu_\alpha}$. 
For simplicity, ${\bf Q}_{0}={\mathbf{id}_{\hilbert}}$ is included in $\mathcal{A}$, 
in a way that the normalization of $\sigma$ is fixed by taking $\langle \mathbf{id}_{\hilbert} \rangle=\mu_0=1$. 
Notice again, however, that ${\bf Q}_0$ is not a true observable.

Using the identity
$$
    \langle i|\delta \log(\sigma)|j\rangle=  \frac{\log(p_i)-\log(p_j)}{p_i-p_j}  \langle i |\delta \sigma|j\rangle,
$$

where $\sigma=\sum_i p_i |i\rangle\langle i|$, it follows that

$$
    \delta S(\sigma) = -\Tr\delta \sigma \log(\sigma)  - \Tr \delta \sigma, 
$$

and hence,

$$
    \left.  -\Tr\left[\delta \sigma \left(\log(\sigma)+ (\lambda_{0}+1)  + \sum_{\alpha=1}^N \lambda_{\alpha} {\bf Q}_{\alpha}\right)\right]\right|_{\sigma=\sigma^\star}=0
$$
As a result, it follows that
$$
\sigma^\star = e^{-{\bf K}-(\lambda_0+1)}=\frac{e^{-{\bf K}}}{\Tr e^{-{\bf K}}},
$$
with ${\bf K}=\sum_{\alpha=1}^N \lambda_{\alpha}{\bf Q}_{\alpha}\in {\cal A}_B$ and $\lambda_0=\log(\Tr e^{-{\bf K}})-1$ to fix the normalization.

\end{proof}

\section{Properties of the KMB product}

This section presents a compilation of properties pertaining to the KMB scalar product.

\begin{lem} \underline{The \textnormal{KMB} scalar product.}
Let $\sigma>0$ be a normalized density operator s.t. $\Tr \sigma=1$.
Then

$$
    ({\bf A},{\bf B})_\sigma^{\rm KMB} = \int_0^1 d\tau \, \Tr \sigma^{1-\tau}{\bf A}^\dagger \sigma^\tau {\bf B},
$$

is a scalar product.
\end{lem}

\begin{proof}
Linearity in ${\bf B}$ and antilinearity in ${\bf A}$ is evident. 
In order to show the positivity, let's consider the following basis on ${\cal A}$, $\{{\bf b}_{ij}/{\bf b}_{ij}=|i\rangle\langle j|\}$ s.t. $\sigma |i\rangle = p_i |i\rangle$.
Then, the Gram's matrix

\begin{eqnarray*}
[{\cal G}_{\sigma(t)}^{\rm KMB}]_{ij,kl}&=&({\bf b}_{ij},{\bf b}_{kl})_{\sigma}^{\rm KMB}=\int_0^{1} d\tau \, \Tr \sigma^{1-\tau} {\bf b}_{ij}^\dagger \sigma^{\tau}{\bf b}_{kl} \\
&=&\int_0^1 d\tau \, \langle j|l\rangle\langle i|k\rangle  p_i (p_j/p_i)^{\tau}\\
&=&\langle j|l\rangle\langle i|k\rangle \frac{p_j-p_i}{\log(p_j/p_i)},
\end{eqnarray*}
is diagonal in this basis, with positive entries $\frac{p_j-p_i}{\log(p_j/p_i)}$, in a way that the associated form is positive definite.
\end{proof}

\begin{lem}\label{lem:KMBandexpectationvalues}
For any $\sigma$, the associated $\textnormal{KMB}$ scalar product satisfies, for any operator ${\bf Q}\in{\cal A}$,

$({\bf id}_\hilbert,{\bf Q})_\sigma=\Tr \sigma {\bf Q}=\langle {\bf Q}\rangle_\sigma$. 

\noindent In particular, if ${\bf Q}={\bf id}_\hilbert$, then $({\bf id}_\hilbert,{\bf id}_\hilbert)=\Tr \sigma=1$.
\end{lem}

\begin{proof}
The proof is straightforward from the definition.
\end{proof}

\begin{lem}\label{lem:KMBisRealSP}
The \textnormal{KMB} scalar product, regarding the state $\sigma$, is a \emph{real} scalar product

$$
({\bf O},{\bf Q})^{\rm KMB}_\sigma = \left(({\bf O}^\dagger,{\bf Q}^\dagger)^{\rm KMB}_\sigma\right)^*.
$$
\end{lem}

\begin{proof}

Using the cyclic property of the trace,

\begin{eqnarray*}
  ({\bf O},{\bf Q})_\sigma^{\rm KMB}&=&\int_0^{1} d\tau \,\Tr(\sigma^{1-\tau}{\bf O}^\dagger \sigma^{\tau}{\bf Q})\\
                                  &=&\int_0^{1} d\tau \,\Tr(\sigma^{\tau}{\bf Q}\sigma^{1-\tau}{\bf O}^\dagger )\\
                                    &=&\int_0^{1} d\tau \,\Tr(\sigma^{1-\tau}{\bf Q}\sigma^{\tau}{\bf O}^\dagger )\\
  &=&({\bf Q}^\dagger,{\bf O}^\dagger)_\sigma^{\rm KMB} =(({\bf O}^\dagger,{\bf Q}^\dagger)_\sigma^{\rm KMB})^*,
\end{eqnarray*}

meaning that $(\cdot,\cdot)_\sigma^{\rm KMB}$ is a \emph{real} scalar product.
\end{proof}

\begin{lem}
Let $(\cdot,\cdot)$ a \emph{real} scalar product,  and  ${\bf A}={\bf A}^\dagger, \,{\bf B}={\bf B}^\dagger\in {\cal A}$. Then, 
$
({\bf A},{\bf B})=({\bf A},{\bf B})^*\in \mathbb{R}.
$
\end{lem}
\begin{proof}
$$({\bf A},{\bf B})=({\bf A}^\dagger,{\bf B}^\dagger)=({\bf A},{\bf B})^* \Rightarrow ({\bf A},{\bf B})\in \mathbb{R}.$$
\end{proof}

\subsection{Orthogonal projections and real scalar products}\label{app:realsp}
That $\pi$ is an orthogonal projection w.r.t. a real scalar product is an important property because, assuming that its image is closed under $\dagger$, it means that $\pi$ maps self-adjoint operators onto self-adjoint operators.  

\begin{lem}\label{lem:selfadjointpi}
Let $\pi$ a linear orthogonal projector w.r.t a real scalar product $(\cdot,\cdot)$ s.t.  $\pi((\pi {\bf A})^\dagger)=(\pi {\bf A})^\dagger$ for any ${\bf A}\in{\cal A}$. 
Then $\pi ({\bf A}^\dagger)=(\pi{\bf A})^\dagger$, for any ${\bf A}\in{\cal A}$.
\end{lem}

\begin{proof}
Every operator in ${\cal A}\in{\cal A}$ admits a decomposition  
\begin{eqnarray}\label{eq:sadecomposition}
{\bf A}&=&{\bf A}_{+}+{\bf i}{\bf A}_{-}\\
{\bf A}_{\pm}&=&\frac{{\bf A}\pm{\bf A}^\dagger}{2\sqrt{\pm 1}}={\bf A}_\pm^\dagger \in{\cal A}\,.
\end{eqnarray} 
Then, using the linearity of $\pi$,
$$
\pi {\bf A}^\dagger = (\pi {\bf A}_+)+{\bf i}(\pi {\bf A}_-)=  (\pi {\bf A})^\dagger \Leftrightarrow \pi {\bf A}_{\pm}=(\pi {\bf A}_{\pm})^\dagger,
$$

Hence, the proof is reduced to show that for ${\bf A}={\bf A}^\dagger \in {\cal A}$, $\pi {\bf A}-(\pi {\bf A})^\dagger=0$ or, using the positivity of the scalar product,
\begin{eqnarray*}
0 &=& (\pi {\bf A}-(\pi {\bf A})^\dagger,\pi {\bf A}-(\pi {\bf A})^\dagger)\\
   &=&(\pi {\bf A},\pi{\bf A})+((\pi {\bf A})^\dagger,(\pi{\bf A})^\dagger) -2 \Re (\pi{\bf A},(\pi {\bf A})^\dagger)\\
   &=&(\pi{\bf A},\pi{\bf A})+((\pi {\bf A}),(\pi{\bf A}))^* -2 \Re (\pi{\bf A},(\pi {\bf A})^\dagger)\\
   &=&2 \Re\left( ({\bf A}, \pi{\bf A})- ({\bf A},\pi(\pi {\bf A})^\dagger)\right),
\end{eqnarray*}
where, in the third line, we use the property of reality of the scalar product to rewrite the second term, and in the last line the orthogonality of $\pi$ regarding the scalar product.
Finally, using the hypothesis $\pi (\pi {\bf A})^\dagger =(\pi {\bf A})^\dagger$,
\begin{eqnarray*}
0 &=& \Re\left( ({\bf A}, \pi{\bf A})- ({\bf A}, (\pi {\bf A})^\dagger)\right)\\
  &=& \Re\left( ({\bf A}, \pi{\bf A})- ({\bf A}^\dagger, \pi {\bf A})^*\right)\\
  &=& \Re\left( ({\bf A}, \pi{\bf A})- ({\bf A}, \pi {\bf A})^*\right)\\
  &=&\Re\left( {\bf i}\Im({\bf A}, \pi{\bf A})\right)=0.
\end{eqnarray*}
\end{proof}
\begin{obs}
The condition $\pi(\pi {\bf A})^\dagger=(\pi {\bf A})^\dagger$ is equivalent to asking that, for any ${\bf A}\in{\cal A}$, 
${\bf A}\in\pi({\cal A}) \Leftrightarrow {\bf A}^\dagger \in\pi({\cal A})$, i.e., ${\cal A}_B=\pi({\bf A})$ is closed under $\dagger$.
\end{obs}

\begin{lem}\label{prop:trivial_projection_in_expect_values}
Let ${\bf Q}\in{\cal A}$ a certain observable, $\rho$ a given state, $B$ a set of independent operators such that ${\bf id}_\hilbert \in {\cal A}_B$ and $\pi_{B,\rho}$ an orthogonal projector regarding the KMB (cover) scalar product associated to $\rho$. Then
\begin{equation}
    \langle \pi_{B,\rho}{\bf Q}\rangle_\rho=\langle {\bf Q}\rangle_\rho.
\end{equation}
\end{lem}
\begin{proof}
Writing the expectation value as a KMB \, \emph{covar} scalar product, and using the orthogonality property of the projector,
\begin{eqnarray*}
        \langle \pi_{B,\rho}{\bf Q}\rangle_\rho&=&
        ({\bf id}_\hilbert, \pi_{B,\rho}{\bf Q})_\rho\\
        &=&(\pi_{B,\rho}{\bf id}_\hilbert, {\bf Q})_\rho\\
        &=&({\bf id}_\hilbert, {\bf Q})_\rho=\langle {\bf Q}\rangle_\rho.
\end{eqnarray*}
\end{proof}

\subsection{KMB scalar product and Heisenberg evolution}

The following lemma is going to be useful in the discussion of when a quantity conserved in the Schr\"odinger dynamics is also conserved in
its restricted counterpart:

\begin{lem}\label{lem:KMBandCommK}
Let ${\bf A},{\bf B},{\bf K} \in{\cal A}$ and $\sigma=\exp(-{\bf K})$ a normalized state s.t. $\Tr \sigma=1$. 
Then,
$({\bf A}^\dagger,[{\bf B},{\bf K}])_{\sigma}^{\rm KMB}=({\bf id}_\hilbert,[{\bf A},{\bf B}])_{\sigma}^{\rm KMB}.$
\end{lem}

This lemma is also useful to show the equivalence between the Schr\"odinger's and the Heisenberg's pictures at the level of the ${\bf K}$ dynamics, as well as to explore the effect of infinitesimal symmetry transformations.

\begin{proof}
From the linearity in ${\bf A}$ and ${\bf B}$, it is enough to show the identity for elements of the orthogonal basis ${\bf b}_{ii'}=|i\rangle\langle i'|$ associated to the eigenvectors of $\sigma=\sum_m p_m |m\rangle\langle m|$. 
Since $\sigma = \exp(-{\bf K})$, ${\bf K} \in {\cal A}$ can be expanded wrt. this basis as ${\bf K}=-\sum_m \log(p_m)|m\rangle\langle m|$. 
Then, if ${\bf A}=|i\rangle\langle j|$, ${\bf B}=|k\rangle\langle l|$,

\begin{equation*}
    \begin{split}
        ({\bf A}^\dagger,&[{\bf B},{\bf K}])_{\sigma}^{\rm KMB} \\
    &= \sum_{m} \bigg(|j\rangle\langle i|, \bigg[|k\rangle\langle l|,-\log(p_m)|m\rangle\langle m|\bigg]\bigg)_{\sigma}^{\rm KMB}\\
    &= ( |j\rangle\langle i|, |k\rangle\langle l|(\log(p_k)-\log(p_l)) )\\
    &= \frac{p_{k}-p_{l}}{\log(p_k/p_l)} \delta_{ik}\delta_{jl}(\log(p_k)-\log(p_l))\\
    &= (p_{k}-p_{l}) \delta_{ik}\delta_{jl}\\
    &= \Tr \sigma |k\rangle\langle l|           |j\rangle\langle i| - \Tr  |k\rangle\langle l| \sigma |j\rangle\langle i|\\
    &=  \Tr \sigma |k\rangle\langle l| |j\rangle\langle i| - \Tr  \sigma |j\rangle\langle i| |k\rangle\langle l| \\
    &=\Tr \sigma\,[{\bf A},{\bf B}]=({\bf id}_\hilbert,[{\bf A},{\bf B}])_{\sigma}^{\rm KMB}.
    \end{split}
\end{equation*}
\end{proof}

\subsection{Spectral norm and induced norm inequalities}\label{sec:KMBinducedNorm}

Along the work, we have considered three different metrics in the space of operators, the spectral norm $\|{\bf A}\|=\max_{|\psi\rangle}\frac{|\langle|{\bf A}|\rangle|}{\langle\psi|\psi\rangle}$, and the norms associated to the KMB and \, \emph{covar} scalar products
$\|{\bf A}\|_{\sigma}^{\stackrel{\rm KMB}{\rm covar}}=\sqrt{({\bf A},{\bf A})_{\sigma}^{\stackrel{\rm KMB}{\rm covar}}}$. Some relations among them are presented here, including the proof of \cref{lem:metricineq}.

\begin{lem}\label{prop:kmbnorm_and_re_approx}
\underline{The \textnormal{KMB} distance as an upper bound} \newline
\noindent \underline{for the Relative Entropy.} Let $\rho_0=\exp(-{\bf K}_0)$ and $\sigma=\exp(-{\bf K}_0-\Delta{\bf K})$ with ${\bf K}_0={\bf K}_0^\dagger, \, \Delta{\bf K}=\Delta{\bf K}^\dagger \in {\cal A}$, s.t. $\Tr \rho=\Tr \sigma=1$. Then, 

\begin{equation}\label{eq:sqrt_re_loc}
    S(\rho_0\|\sigma)\leq \|\Delta{\bf K}\|_{\rho_0}^{\rm KMB}.
\end{equation}

\end{lem}
\begin{proof}
The relative entropy between $\rho$ and $\sigma$ can be written as the expectation value of ${\bf A}$:

\begin{equation}\label{eq:ReAsaKMBprod}
S(\rho_0\|\sigma)=
\Tr \rho ((-{\bf K}_0)-(-{\bf K}_0-\Delta{\bf K}))=\Tr \rho_0 \Delta{\bf K},
\end{equation}
which, by \cref{lem:KMBandexpectationvalues} can be written as 
$S(\rho\|\sigma)=({\bf id}_\hilbert,\Delta{\bf K})_{\rho_0}$. 
Then, by the Cauchy-Schwarz inequality
$$
S(\rho_0\|\sigma)=({\bf id}_\hilbert,\Delta{\bf K})_{\rho_0}\leq \|{\bf id}_\hilbert\|^{\rm KMB}_{\rho_0} \|\Delta{\bf K}\|^{\rm KMB}_{\rho_0} =\|\Delta{\bf K}\|^{\rm KMB}_{\rho_0},
$$
since $(\|{\bf id}_\hilbert\|^{\rm KMB}_{\rho_0})^2=\Tr \rho_0 {\bf id}_{\hilbert}=1$.
\end{proof}

Now, we are in conditions to show \cref{lem:metricineq}:
  \begin{proof}
    The first inequality is a direct consequence of the spectral norm definition and the spectral decomposition of $\sigma=\sum_k p_k |k\rangle\langle k|$:
    \begin{eqnarray*}
      (||{\bf A}||^{{\rm covar}}_{\sigma})^2 &=&\Tr \sigma \textnormal{ } \frac{{\bf A}^\dagger{\bf A}+{\bf A}{\bf A}^\dagger}{2}\\
                             &=&\sum_{k}p_k \frac{\langle k|{\bf A}^\dagger{\bf A}|k\rangle}{2} +
                                \sum_{k}p_k \frac{\langle k|{\bf A}{\bf A}^\dagger|k\rangle}{2}\\
                             &\leq &\sum_k p_k \frac{\|{\bf A}\|^2+\|{\bf A}^\dagger\|^2}{2}\\
                             &=& \sum p_k \|{\bf A}\|^2 = \|{\bf A}\|^2.
      \end{eqnarray*}
      
For the second inequality, it is enough to show that ${\cal G}^{\textnormal{KMB}}_{\sigma} \leq {\cal G}^{{\rm covar}}_{\sigma}$. 
Also, as both matrices are diagonal in the canonical basis associated with the eigenvectors of $\sigma$, it is enough to show that

      \begin{eqnarray}
      \bigg(|| \textnormal{ }\ket{i} \bra{j} \textnormal{ }||^{\sigma}_{\textnormal{KMB} }\bigg)^2&\leq & \bigg(|| \textnormal{ }\ket{i} \bra{j} \textnormal{ }||^{\sigma}_{{\rm covar} }\bigg)^2.    
      \end{eqnarray}
or, in terms of the eigenvalues $p_i$ of $\sigma$,

      $$
      \frac{(p_i-p_j)/\log(p_i/p_k)}{(p_i+p_j)/2}\leq 1.
      $$
Now, without loss of generality, let's assume that $p_i\geq p_j$ and hence, $\log(p_i/p_j)=x \geq 0$. Then,
    \begin{eqnarray*}
        \frac{(p_i-p_j)/\log(p_i/p_k)}{(p_i+p_j)/2} &=&\frac{1-e^{-x}}{x/2 (1+e^{-x}) }\\
        &=&   \frac{\tanh(x/2)}{x/2}\leq 1.     
      \end{eqnarray*}

Finally, the last inequality follows directly from \cref{prop:kmbnorm_and_re_approx}. 
      
\end{proof}

\section{KMB geometry and series expansions}\label{app:KMBgeom}

After having introduced the fundamental properties of the KMB product, and its induced distance, we are in a position to use its definition to write the first non-trivial orders in the series expansion of the relevant quantities of this work, like expectation values and relative entropies.

\subsection{Proof of \cref{prop:kubo1}}\label{app:proof_kubosp}
\begin{proof}

Replacing $\rho_{\lambda}$ in the L.H.S. of \cref{eq:KWBvar} by its series expansion around $\lambda$ yields 

\begin{equation}\label{eq:seriesrhoip}
\rho_{\lambda + \delta \lambda} = \rho_{\lambda} ({\bf id}_\hilbert+\delta \lambda \int d\tau \, \rho_\lambda^{1-\tau}\Delta {\bf K} \rho_{\lambda}^{\tau} + {\cal O}(\delta\lambda^2) ).
\end{equation}

It follows that

$$
\Tr \rho_{\lambda + \delta \lambda} {\bf O}=\Tr \rho_{\lambda}{\bf O} + \delta \lambda \Tr \int d\tau \, \rho_\lambda^{1-\tau}\Delta {\bf K} \rho_{\lambda}^{\tau}{\bf O} + {\cal O}(\delta \lambda^2),
$$

\noindent or, in terms of the KMB scalar product regarding $\rho_\lambda$ \cref{eq:KWBdef1}, 

\begin{equation}\label{eq:seriesav}
\Tr \rho_{\lambda + \delta \lambda} {\bf O}=({\bf id}_\hilbert,{\bf O})^{\rm KMB}_{\rho_\lambda} + \delta \lambda (\Delta{\bf K}^\dagger,{\bf O})^{\rm KMB}_{\rho_\lambda} + {\cal O}(\delta \lambda^2), 
\end{equation}
Notice that the hermiticity of ${\bf K}_0$ is required to ensure that  $\rho_0>0$ and $(\cdot,\cdot)_{\rho_0}^{\rm KMB}$ be an scalar product.
\noindent the RHS can be read then from the linear term. 
\end{proof}

\subsection{The $\pi_B$ projector is an orthogonal projector w.r.t. the KMB scalar product}\label{sec:proofOrtProj}

\begin{proof}\label{proof:OrtProj}

To show the statement, it is enough to show the equivalent statement that for any ${\bf Q}=\pi_{B,\rho_0}{\bf Q}\in {\cal A}_B$
$$
({\bf Q},\Delta {\bf K})_{\rho_0}^{\rm KMB}=({\bf Q},\pi_{B,\rho_0}\Delta {\bf K})_{\rho_0}^{\rm KMB}
$$

To show it, lets start by defining the (non-normalized) state $\rho_\lambda=\exp(-{\bf K}_0+\lambda \Delta {\bf K})$ s.t. $\Tr \rho_0=1$ and $\rho_0\in {\cal M}_B$.
Using \cref{prop:kubo1} and the condition \cref{eq:nlProj},
\begin{eqnarray*}
({\bf Q},\Delta {\bf K})_{\rho_0}^{\rm KMB}&=&\frac{\partial}{\partial \lambda}\Tr \mathbf{Q}\rho_\lambda\\
&=&\frac{\partial}{\partial \lambda}\Tr \mathbf{Q}({\cal P}_B\rho_\lambda)\\
&=&\frac{\partial}{\partial \lambda}\Tr \mathbf{Q}\exp(-\Pi_B ({\bf K}_0 - \lambda \Delta{\bf K}) )
\end{eqnarray*}
where in the last line we have used the definition of $\Pi_B$ \cref{eq:compatconddiffbis}. Now, using again \cref{prop:kubo1},
$$
\Tr \mathbf{Q}\exp(-{\bf K}_0+ \lambda \Delta{\bf K}' )=(\mathbf{Q}, \lambda \Delta{\bf K}' + {\cal O}(\lambda^2) )_{\rho_0}^{\rm KMB}\\
$$
\noindent with $\Delta{\bf K}'=\frac{\partial}{\partial \lambda}(\Pi_B (\lambda \Delta{\bf K}-{\bf K}_0)-{\bf K}_0)$. Finally, using the definition \cref{eq:compatconddiff},
$$
\Delta{\bf K}'=\pi_{B,\rho_0}\Delta{\bf K}
$$
and therefore,
$$
({\bf Q},\Delta {\bf K})_{\rho_0}^{\rm KMB}=(\mathbf{Q}, \pi_{B,\rho_0} \Delta{\bf K})_{\rho_0}^{\rm KMB}.
$$
\end{proof}

\subsection{KMB product and Series Expansion for the Relative Entropy}\label{sec:KMBRelEntr}

\begin{lem}\label{lem:seriesTr}
\underline{Taylor's series of $\Tr \rho_\lambda$.} 
Let $\rho_\lambda=\exp(-{\bf K}+\lambda \Delta {\bf K})$. Then

\begin{equation}\label{eq:seriestrrho} 
\begin{split}
    \Tr \rho_{\lambda+\delta \lambda} &= \Tr \rho_{\lambda} + \delta \lambda ({\bf id}_\hilbert,{\bf \Delta K})^{\rm KMB}_{\rho_\lambda} \\
    &\,\, + \frac{\delta \lambda^2}{2}(\Delta{\bf K}^\dagger,\Delta{\bf  K})^{\rm KMB}_{\rho_{\lambda}}+{\cal O}(\delta \lambda^3).
\end{split}
\end{equation}
\end{lem}

\begin{proof}
By tracing out on both sides of \cref{eq:seriesrhoip}, the well-known relation $\frac{d}{d\lambda}{\rm Tr}\rho_\lambda = \Tr \rho_\lambda \Delta {\bf K}$ follows.
Therefore, the term of order $k$ in \cref{eq:seriesav} is linked to the $(k+1)$-th term, multiplied by $k+1$, in \cref{eq:seriesav}.
\end{proof}

\begin{lem}\label{prop:kmbnorm_and_re}
\underline{\textnormal{KMB} distance as an approximation to} 
\noindent \underline{the Relative Entropy.} Let 
$\rho_0=\exp(-{\bf K}_0)$ s.t. $\Tr \rho_0=1$, and  $\rho_\lambda=\frac{\exp(-{\bf K}_0-\lambda {\bf B})}{\Tr \exp(-{\bf K}_0-\lambda {\bf B})}$ and
$\sigma_\lambda=\frac{\exp(-{\bf K}_0-\lambda {\bf A})}{\Tr \exp(-{\bf K}_0-\lambda {\bf A})}$ 
with ${\bf A}, {\bf B} \in {\cal A}$ s.t. $\Tr\mathbf{A}\rho_0=\Tr\mathbf{B}\rho_0=0$. 
Then, 
  
\begin{eqnarray}\label{eq:sqrt_re_loc_ub}
    S(\rho_\lambda\|\sigma_\lambda)&=& 
        \lambda^2\frac{( \|{\bf A}-{\bf B}\|_{\rho_0}^\textnormal{KMB})^2}{2}  + {\cal O}(\lambda^3).
\end{eqnarray}
\end{lem}

\begin{proof}
Using the identity \cref{eq:ReAsaKMBprod}, 

\begin{equation}\label{eq:REfulllambda}
S(\rho_\lambda\|\sigma_\lambda)=({\bf id}_\hilbert,\Delta{\bf K})_{\rho_\lambda}^{\rm KMB}=\Tr \Delta{\bf K}\rho_\lambda,
\end{equation}

\noindent with

\begin{eqnarray*}
  \Delta{\bf K}&=&\log \rho_\lambda-\log\sigma_\lambda \\&=&\lambda({\bf A}-{\bf B})+\log\frac{\Tr \exp(-{\bf K}_0-\lambda{\bf A}) }{\Tr \exp(-{\bf K}_0-\lambda{\bf B})}{\bf id}_\hilbert.
\end{eqnarray*}

Next, we use a second-order Taylor expansion in $\lambda$ around $\lambda = 0$.
Using the condition $\Tr{\bf A}\rho_0=\Tr{\bf B}\rho_0=0$
and \cref{lem:seriesTr}, the logarithm of the quotient of traces in $\Delta {\bf K}$ is given by

$$
{\small \log\frac{\Tr e^{-{\bf K}_0-\lambda{\bf A}} }{\Tr e^{-{\bf K}_0-\lambda{\bf B}}}=
\lambda^2\frac{(\|{\bf A}\|^{\rm KMB}_{\rho_0})^2-(|{\bf B}\|_{\rho_0}^{\rm KMB})^2}{2}
+{\cal O}(\lambda^3)}.
$$

Since $\Delta {\bf K}\approx {\cal O}(\lambda)$, the second order expansion of \cref{eq:REfulllambda} is obtained by expanding the last member up to first order in $\lambda$ for $\Delta {\bf K}$ fixed, and replacing then $\Delta {\bf K}$ by its second order expansion, thus yielding

\begin{eqnarray*}
S(\rho_\lambda\|\sigma_\lambda)&=&\Tr \Delta{\bf K}\rho_0 + (\Delta{\bf K},\Delta {\bf K})_{\rho_0}^{\rm KMB}+{\cal O}(\lambda^3)\\
&=&\lambda^2\frac{(\|{\bf A}\|^{\rm KMB}_{\rho_0})^2-(|{\bf B}\|_{\rho_0}^{\rm KMB})^2}{2}+\\
&&
\lambda^2({\bf B},{\bf A}-{\bf B})_{\rho_0}^{\rm KMB}+{\cal O}(\lambda^3)\\
&=&\frac{\lambda^2( \|{\bf A}-{\bf B}\|_{\rho_0}^{\rm KMB})^2}{2}+{\cal O}(\lambda^3).
\end{eqnarray*}
\end{proof}

\begin{obs}
By rescaling, it is possible to identify the expansion parameter $\lambda$ with the KMB norm of the difference $\|{\bf A}-{\bf B}\|_{\rho_0}^{\rm KMB}$, giving the same asymptotic behavior.
\end{obs}

\begin{obs}
As mentioned in \cref{section_max_ent}, the Relative Entropy is a measure of indistinguishability of two states, and therefore,
the KMB distance has the same role for asymptotically close states. Moreover, \cref{prop:kubo1} and the Cauchy-Schwarz inequality provide another operational interpretation of this metric. Suppose that we want to discriminate two close states $\rho_0=\exp(-{\bf K}_0)$ and  $\rho_\lambda=\exp(-{\bf K}_0+\lambda \Delta {\bf K})$ by looking at the expectation value of some operator ${\bf Q}$. 
From \cref{prop:kubo1},  $\Delta \langle {\bf Q}\rangle\approx |(\lambda \Delta {\bf K},{\bf Q})_{\rho_0}^{\rm KMB}|+{\cal O}(\lambda^2)$. A rough estimation of how difficult is to distinguish both states through this measurement is by comparing $\Delta  \langle{\bf Q}\rangle$ with
$\sigma_{\bf Q}=\sqrt{\langle{\bf Q}^2\rangle}\approx \|{\bf Q}\|_{\rho_0}^{\rm (cov)}>\|{\bf Q}\|_{\rho_0}^{\rm KMB}$. Then, using the Cauchy-Schwarz inequality
\begin{eqnarray}
\frac{\Delta {\bf Q}}{\sqrt{\langle{\bf Q}^2\rangle}} &\leq& \frac{\|{\bf Q}\|_{\rho_0}^{\rm KMB}}{\|{\bf Q}\|_{\rho_0}^{\rm (cov)}}
\|{\bf K}\|_{\rho_0}^{\rm KMB}\leq \|{\bf K}\|_{\rho_0}^{\rm KMB}
\end{eqnarray}
where we used \cref{lem:metricineq} to eliminate the dependence of ${\bf Q}$ in the last member. 
Despite a more careful analysis of the role of these norms 
in state-discrimination tasks is beyond the scope of this work,
the previous argument is enough to say that the 
KMB-norm (as well as the\, \emph{covar}-norm) is a measure of how difficult it is to discriminate between two states by comparing expectation values. 
\end{obs}

\begin{obs}
The relative entropy is asymptotically symmetric for states close to each other.
\end{obs}
\begin{obs}
Different from \cref{prop:kmbnorm_and_re_approx}, \cref{prop:kmbnorm_and_re} is a statement valid beyond the limit $\|\Delta{\bf K}\|_{\rho_0}^{\rm KMB}\rightarrow 0$, being valid always $\Tr \exp(-{\bf K}_0-\Delta{\bf K})<\infty$.
Nevertheless, it makes sense to check that both statements are equivalent in the asymptotic limit.
To see this, it is convenient to identify $\Delta{\bf K}\equiv \Delta{\bf K}_1$ in \cref{prop:kmbnorm_and_re} with  

$$
\Delta{\bf K}_\lambda=\lambda{\bf A}+{\bf id}_\hilbert \log(\Tr\exp(-{\bf K}_0-\lambda {\bf A})),
$$

with $\Tr \rho_0{\bf A}=0$.
Assuming ${\bf A}$ \emph{small}, $\Delta{\bf K}$ is too and viceversa, and from \cref{eq:seriestrrho},

$$
\Delta {\bf K}_\lambda= \lambda {\bf A}+ \lambda^2 \frac{\|{\bf A}\|^2}{2}{\bf id}_\hilbert+{\cal O}(\lambda^3),
$$

and hence,

$$
\|\Delta {\bf K}_\lambda\|^{\rm KMB}_{\rho_0}= \|\lambda {\bf A}\|^{\rm KMB}_{\rho_0}+ {\cal O}(\lambda^3),
$$

in a way that in the asymptotic limit, the conditions $\Tr \rho_0{\bf A}=0$ and $\Tr \rho_0\exp( -{\bf K}_0-\lambda {\bf A})=1$  are equivalent.
\end{obs}

\section{Schr\"odinger equation on ${\bf K}$}\label{sec:proofSchK}

\cref{lem:k_scr_dyn} establishes that if $\rho(t)$ is a solution of \cref{eq:Schrodinger} then ${\bf K}(t)=-\log \rho(t)$ is too. 

\begin{proof}
To see this, we observe that the solution of \cref{eq:Schrodinger} can be written as 
$$
\rho(t)={\bf U}(t)^\dagger e^{-{\bf K}(0)}{\bf U}(t)=e^{-{\bf U}(t)^\dagger{\bf K}(0){\bf U}(t)}=e^{-{\bf K}(t)}
$$
with ${\bf K}(t)={\bf U}(t)^\dagger{\bf K}(0){\bf U}(t)$ and ${\bf U}(t)$ a unitary operator, solution of the equation
$$
{\bf i}\hbar\frac{d}{dt}{\bf U}={\bf H} {\bf U}
$$
with initial condition ${\bf U}(0)={\bf id}_\hilbert$. But then, 
$$
\frac{d{\bf U}}{dt}{\bf U}^\dagger = -\frac{d{\bf U}^\dagger }{dt}{\bf U}=\frac{{\bf H}}{{\bf i}\hbar}
$$
and hence, 
$$
\frac{d {\bf K}}{dt}=\frac{d{\bf U}}{dt}{{\bf U}}^\dagger(t){\bf K}(t) + {\bf K}(t) {\bf U}(t)  \frac{d{\bf U}^\dagger}{dt}=\frac{[{\bf H},{\bf K}]}{{\bf i}\hbar}
$$
satisfies \cref{eq:SchrodingerK} with ${\bf K}_0=-\log (\rho(0))$ as initial condition.
\end{proof}

\section{Properties of the restricted dynamics}

\subsection{Conserved quantities of the restricted dynamics}

The Schr\"odinger equation on $\rho$, given by \cref{eq:Schrodinger}, and the Schr\"odinger equation on ${\bf K}$, given by $\cref{eq:SchrodingerK}$, are completely equivalent, since the mapping $\exp: {\cal A} \rightarrow {\cal S}(\hilbert)$ has an-everywhere-well-defined inverse mapping.
As such, they share the same dynamical properties, eg. the conservation of the von Neumann entropy amongst others. 
It is not clear \textit{a priori} which of these dynamical properties hold for the restricted dynamics as well.

From \cref{prop:kubo1} it follows

\begin{prop}\label{prop:entropy_and_normalization_preserving}
\underline{Conservation quantities of interest}
Let $B$ a basis of operators s.t. ${\bf id}_\hilbert\in B$ and $\rho(t)=\exp(-\tilde{\bf K}_B(t))$, with 
$\tilde{{\bf K}}_B(t)$ a solution of \cref{eq:restricted} s.t. $\Tr \rho(0)=1$. Then
\begin{enumerate}
\item $\Tr \rho(t)=\Tr \rho(0) \, \forall t,$ 
\item $S(\rho(t))=S(\rho(0))$, with $S(\rho)$ the von-Neumann entropy \cref{eq:vnEntropy}.
\end{enumerate}
\end{prop}

\begin{proof}
Let's start by noticing that since $[\rho(t),\tilde{\bf K}_B(t)]=0$,
$$
\langle [{\bf H},\tilde{\bf K}_B]\rangle_{\rho(t)}=({\bf id},[{\bf H},\tilde{\bf K}_B])_{\rho}^{\rm KMB}=\Tr \rho(t)[{\bf H},\tilde{\bf K}_B] = 0.
$$

Then, the trace-preserving property follows from
\begin{eqnarray*}
    \frac{d \Tr \rho}{dt}
        &=&({\bf id}_\hilbert,-\frac{\partial \tilde{\bf K}_B}{\partial t})_{\rho}^{\rm KMB}\\
        &=& \left\langle -\frac{\partial \tilde{\bf K}_B}{\partial t}  \right\rangle_\rho\\
        &=&\left\langle \pi_B \frac{-[{\bf H},\tilde{\bf K}_B]}{{\bf i} \hbar} \right\rangle_\rho\\
        &=& \frac{-\langle[{\bf H},\tilde{\bf K}_B]\rangle_\rho}{{\bf i} \hbar} 
        =0
\end{eqnarray*}

On the other hand, the conservation of the von Neumann entropy follows by noticing that

$$
S(\rho)=\Tr \tilde{\bf K}_B \rho=\langle \tilde{\bf K}_B\rangle_\rho,
$$

with a time derivative given by
$$
\frac{d}{dt}\langle \tilde{\bf K}_B\rangle_\rho=
\left\langle \frac{d}{dt}\tilde{\bf K}_B \right\rangle_\rho-
\left(\tilde{\bf K}_B,\frac{d\tilde{\bf K}_B}{dt}\right)^{\rm KMB}_\rho,
$$

The first term, which comes from the change of ${\bf K}$, can be rewritten as

\begin{eqnarray*}
\left\langle \frac{d}{dt}\tilde{\bf K}_B \right\rangle_\rho&=&
\left\langle \pi_B \frac{[{\bf H},\tilde{\bf K}_B]}{{\bf i}\hbar} \right\rangle_\rho\\
&=&\left\langle [{\bf H},\tilde{\bf K}_B] \right\rangle_\rho=0,
\end{eqnarray*}

while the second term, coming from the change in $\rho$, 
\begin{eqnarray*}
\left(\tilde{\bf K}, \frac{d\tilde{\bf K}}{dt}\right)^{\rm KMB}_\rho&=&
\left(\tilde{\bf K}, \pi_B\frac{[{\bf H},\tilde{\bf K}]}{{\bf i}\hbar}\right)^{\rm KMB}_\rho\\
&=&\left(\pi_B\tilde{\bf K}, \frac{[{\bf H},\tilde{\bf K}]}{{\bf i}\hbar}\right)^{\rm KMB}_\rho \\
%\!\!
&=&\left(\tilde{\bf K}, \frac{[{\bf H},\tilde{\bf K}]}{{\bf i}\hbar}\right)^{\rm KMB}_\rho\\
&=& \int d\tau \, \frac{\Tr \rho^{1-\tau}\tilde{\bf K}\rho^\tau[{\bf H},\tilde{\bf K}]  }{{\bf i}\hbar} d\tau\\
&=&\int d\tau \, \frac{\Tr \rho\tilde{\bf K} [{\bf H},\tilde{\bf K}]  }{{\bf i}\hbar} =0
\end{eqnarray*}

\end{proof}

It follows that, in general, if $\pi_B{\bf Q}={\bf Q}$ (namely, $\bf Q \in {\cal A}_B$), then the restricted evolution of its expectation value follows a free Ehrenfest evolution:

\begin{prop}\label{prop:ehrenfest}
Let $\rho(t)=\exp(-\tilde{\bf K}_B(t))$ with
$\tilde{\bf K}_B(t)$ a solution of \cref{eq:restricted},  $\pi_B$ a \textnormal{KMB} orthogonal projector regarding $\rho(t)$ over a subspace ${\cal A}_B \subset {\cal A}$, and 
${\bf Q}\in{\cal A}$ s.t. $\pi_B{\bf Q}={\bf Q}$. 
Then
$$
\frac{d}{dt}\Tr \rho {\bf Q}=\Tr \rho [{\bf Q},{\bf H}]
$$
In particular, if $[{\bf Q},{\bf H}]=0$, $\Tr \rho {\bf Q}$ is a constant of motion for the restricted evolution.
\end{prop}

\begin{proof}
Using \cref{prop:kubo1} 
$$
\frac{d}{dt}\Tr \rho {\bf Q}=\left({\bf A}^\dagger, \pi_B\frac{[{\bf H},{\bf K}]}{{\bf i}\hbar}\right)^{\rm KMB}_\rho.
$$

and from \cref{prop:ortproj},

$$
\left({\bf A}^\dagger, \pi_B\frac{[{\bf H},{\bf K}]}{{\bf i}\hbar}\right)^{\rm KMB}_\rho=
\left(\pi_B{\bf A}^\dagger, \frac{[{\bf H},{\bf K}]}{{\bf i}\hbar}\right)^{\rm KMB}_\rho.
$$

Then, from \cref{prop:ehrenfest}, and using the property $\pi_B{\bf A}^\dagger=(\pi_B{\bf A})^\dagger={\bf A}^\dagger$,

$$
\left({\bf A}^\dagger, \pi_B\frac{[{\bf H},{\bf K}]}{{\bf i}\hbar}\right)^{\rm KMB}_\rho=
\left({\bf id}_\hilbert, \frac{[{\bf A},{\bf H}]}{{\bf i}\hbar} \right)^{\rm KMB}_\rho=
\Tr \rho \frac{[{\bf A},{\bf H}]}{{\bf i}\hbar}.
$$
\end{proof}

As a result, $\pi_B{\bf K}$ provides an explicit approximate solution for the Max-Ent optimization problem \cref{eq:REPDef}, converging to the exact solution provided ${\bf K}$ is close-enough, in the sense of the KMB distance, to some ${\bf K}' \in {\cal A}$.

\subsection{Error estimations}\label{app:sccriteria}

In \cref{sec:restricteddynamics}, we studied the problem of estimating the errors introduced when approximating the projected dynamics by the restricted dynamics -- see \cref{eq:restricted}. 
This led us to introduce $\tilde{\Delta}$~--\cref{eq:deftildeDelta}-- and $\Delta$ ~--\cref{eq:defDelta}--, which depend on the difference $\Delta {\bf K}$ ~--\cref{eq:defDeltaK}-- between the solutions ${\bf K}(t)$ of the \emph{free} Schr\"odinger equation ~--\cref{eq:SchrodingerK}-- and $\tilde{\bf K}_B(t)$  of the restricted evolution.
In the general case of large many-body systems, we only have access to $\tilde{\bf K}_B(t)$ (or at least, some kind of approximation to it), but not for ${\bf K}(t)$, which requires an exponentially large number of parameters. 
To estimate $\Delta {\bf K}$, the idea is to build an integral equation for it, in terms of the given solution $\tilde{\bf K}_B(t)$, the Hamiltonian of the original system ${\bf H}$, and the projections $\pi_{t}\equiv \pi_{B,\exp(-\tilde{\bf K}_B(t))}$.
To do that, we start by noticing that

$$
\frac{d\Delta{\bf K}}{dt}=\frac{d{\bf K}}{dt}-\frac{d\tilde{\bf K}_B}{dt}=\frac{1}{{\bf i}\hbar}\left([{\bf H},{\bf K}]-\pi_t[{\bf H},\tilde{\bf K}_B]\right),
$$
which, by adding and subtracting $[{\bf H},\tilde{\bf K}_B]/({\bf i}\hbar)$, can be rewritten as

\begin{equation}\label{eq:diffeqDeltaK}
\frac{d\Delta{\bf K}}{dt}=\frac{1}{{\bf i}\hbar}\left([{\bf H},\Delta{\bf K}]+\pi^{\perp}_t[{\bf H},\tilde{\bf K}_B]\right),
\end{equation}

with $\pi_t^{\perp}={\bf id}_\hilbert-\pi_t$ is the instantaneous linearized projection onto the orthogonal space to ${\cal A}_B$, in the neighborhood of $\tilde\sigma(t)=\exp(-\tilde{\bf K}_B(t))$.
Hence, $\Delta {\bf K}(t)$ is the solution of a linear differential equation, with an inhomogeneity controlled by $\tilde{\bf K}_B(t)$, which at any $t$ lies in the orthogonal complement to ${\cal A}_B$.
Now, using that $\Delta{\bf K}(0)=0$, we can rewrite the \cref{eq:diffeqDeltaK} as a Volterra's second kind equation,

\begin{eqnarray}
\Delta {\bf K}&=&
\int_0^t dt' \,  \frac{\pi_{t'}^\perp [{\bf H},\tilde{\bf K}_B]}{{\bf i}\hbar}+
\int_0^t dt' \,  \frac{[{\bf H},\Delta{\bf K}]}{{\bf i}\hbar},
\end{eqnarray}

with formal solution

\begin{eqnarray}
\Delta {\bf K}&=&\sum_{m=0}^\infty \Delta{\bf K}_m(t)\\
\Delta{\bf K}_0(t)&=&\int_0^t dt' \, \frac{\pi_{t'}^\perp [{\bf H},\tilde{\bf K}_B]}{{\bf i}\hbar}\\
\Delta{\bf K}_{m+1}(t)&=&\int_0^t dt' \, \frac{[{\bf H},\Delta{\bf K}_{m}(t')]}{{\bf i}\hbar}.
\end{eqnarray}

The first term ${\bf K}_0(t)$ can be kept small by choosing a suitable basis $B$, while for large but finite-dimensional systems, (any) norm of $\Delta{\bf K}(t)$ can be bounded by $M t \max_t\|\Delta{\bf K}(t)\|$, for a certain positive constant $M$. 
In the short time limit $t\rightarrow 0$, the leading order of the expansion is given by ${\bf K}_0(t)$ which, for a time-independent Hamiltonian and hierarchical basis $B\equiv B_{\ell}$, defined in ~\cref{eq:iterated_comm_basis}, grows as $t^{\ell+1}$. 
Keeping this term, $\tilde\Delta(t)$ can be estimated by 

\begin{eqnarray*}
\tilde\Delta(t)&=&\left\|\Delta {\bf K}(t)\right\|^{\rm KMB}_{\sigma(t)}\approx \|\Delta{\bf K}\|^{\rm KMB}_{\rho(t)}\\
&\geq&\|{\bf K}-\Pi_B{\bf K}\|^{\rm KMB}_{\rho(t)}\approx \|{\bf K}-\pi_{B,\rho(t)}{\bf K}\|^{\rm KMB}_{\rho(t)}\nonumber,
\end{eqnarray*}

which provides a criteria for the correctness of using $\tilde{\bf K}_B(t)$ as an approximation of ${\bf K}(t)$. 
On the other hand, if the goal is to approximate $\Pi_{B}{\bf K}(t)$,

\begin{eqnarray*}
\Delta(t) &=& \|\tilde{\bf K}_B-\Pi_B{\bf K}\|^{\rm KMB}_{\sigma(t)}\\
&\approx&
\|\tilde{\bf K}_B-\pi_{B,\sigma(t)}{\bf K}\|^{\rm KMB}_{\sigma(t)}\\
          &=& \left\|\pi_{B,\sigma(t)} \Delta{\bf K}(t)\right\|^{\rm KMB}_{\sigma(t)}\\
&\leq &\tilde\Delta(t).
\end{eqnarray*}

bounds the errors incurred by the approximation. 

\subsection{Convergence and Hierarchical Basis}\label{sec:apphierarch}

At the end of \cref{sec:projected_and_constrained_dynamics}, it was introduced the notion of \emph{Hierarchical basis} to discuss how the projected dynamics converges to the one obtained from the restricted evolution, as the basis defining ${\cal M}_B$ is enlarged by adding new relevant operators. 

\begin{lem}\label{lem:pertproj_ev}
  Let ${\bf K}(t)$ a solution of \cref{eq:SchrodingerK} with a time-independent Hamiltonian ${\bf H}$, and ${\bf b}_l$ the sequence of iterated commutators defined in \cref{eq:iterated_comm_basis}. Then
  \begin{equation}
    \|\kern-0.25ex|{\bf K}(t)-\pi_\ell {\bf K}(t)\|\kern-0.25ex| \approx {\cal O}^{\ell+1}(t).
  \end{equation}
  for \underline{any} operator norm $\|\kern-0.25ex|\cdot\|\kern-0.25ex|$ defined over ${\cal A}$.
  \end{lem}

\begin{proof}

Just like \cref{eq:SchrodingerK}, \cref{eq:SchrodingerK} can be solved in the form of a Dyson's series

$$
{\bf K}(t)=\sum_m {\bf K}_m(t),
$$
with 

$${\bf K}_0={\bf K}(0)\;\; \mbox{and}\;\; {\bf K}_{m+1}(t)=\frac{1}{{\bf i}\hbar}\int_0^t [{\bf H},{\bf K}_m(t')]dt'\;.$$

For a time-independent ${\bf H}$, this implies that

$${\bf K}_{m}(t)= \frac{t^m}{m!}{\bf b}_m.$$

Therefore,
\begin{eqnarray*}
  \pi_{\ell}{\bf K}(t)&=& \sum_{m}\frac{t^m}{m!} \pi_{\ell}{\bf b}_l\\
&=&  \sum_{m}\frac{t^m}{m!} ({\bf b}_m-{\bf b}_m  + \pi_{\ell}{\bf b}_m)\\
&=& {\bf K}(t) +  \sum_{m>\ell}\frac{t^m}{m!} (\pi_{\ell}{\bf b}_m-{\bf b}_m)\\
&=&{\bf K}(t) + {\cal O}^{\ell + 1}(t),
\end{eqnarray*}
meaning that
$$
\|\kern-0.25ex|{\bf K}(t)-\pi_{\ell}({\bf K}(t))\|\kern-0.25ex| \sim {\cal O}^{\ell + 1}(t).
$$
for \emph{any} operator norm $\|\kern-0.25ex|\cdot\|\kern-0.25ex|$ defined over ${\cal A}$.
\end{proof}

In a similar fashion,

\begin{lem}\label{lem:pertrestr_ev}
  Let ${\bf K}(t)$ a solution of \cref{eq:SchrodingerK} with a time-independent Hamiltonian ${\bf H}$, ${\bf b}_l$ the sequence of iterated commutators defined in \cref{eq:iterated_comm_basis}, and $\tilde{K}_B(t)$ the solution of \cref{eq:restricted} with $\pi_B\equiv\pi_{B_\ell}$ s.t. $\tilde{\bf K}_B(0)={\bf K}(0)$. Then
  \begin{equation}
    \|\kern-0.25ex|{\bf K}(t)- \tilde{\bf K}_B(t)\|\kern-0.25ex| \approx {\cal O}^{\ell+1}(t). 
  \end{equation}
  for any operator norm $\|\kern-0.25ex|\cdot\|\kern-0.25ex|$ defined over ${\cal A}$.
\end{lem}

\begin{proof}
  The solution of \cref{eq:restricted} can be spanned as  $\tilde{\bf K}_B(t)=\sum_m \tilde{\bf K}_m(t)$ with
  $$
  \tilde{\bf K}_m(t)=\pi_B([{\bf H}, \tilde{\bf K}_{m-1}(t)])=\frac{t^{m}}{m!}\pi_\ell{\bf b}_m.
  $$

Since $\pi_\ell{\bf b}_m={\bf b}_m$ for $m\leq \ell$,

  $$
    {\bf K}(t)-\tilde{\bf K}_B=\sum_{m>\ell}\frac{t^m}{m!}({\bf b}_m-\pi_\ell{\bf b}_m)\approx {\bf O}^{\ell +1}(t).
    $$
Therefore, $\|\kern-.25ex|{\bf K}(t)-\tilde{\bf K}_B\|\kern-.25ex|\approx {\bf O}^{\ell+1}(t)$.
\end{proof}

\begin{cor}
 $ \|\kern-0.25ex|  \pi_\ell {\bf K}(t)-\tilde{\bf K}_{B_\ell}(t) \|\kern-0.25ex|\approx {\cal O}^{\ell+1}(t)$.
\end{cor}
 \begin{proof}Follows from  \cref{lem:pertproj_ev}, \cref{lem:pertrestr_ev}  and the triangular inequality.\end{proof}

\section{KMB and Correlation scalar products in the Gaussian case}

If $\sigma$ is a (bosonic or fermionic) Gaussian state, it is possible to choose a basis for the quadratic forms on creation and annihilation operators, ${\bf a}_i$, ${\bf a}^\dagger_i$, 
satisfying canonical commutation (anticommutation) relations $[{\bf a}_j,{\bf a}_j]_{\pm}=[{\bf a}_i^\dagger,{\bf a}_j]_{\pm}=0$ and 
$[{\bf a}_i,{\bf a}^\dagger_j]_{\pm}=\delta_{ij}$ s.t.

\begin{equation*}
    \begin{split}
        & \langle {\bf a}_i\rangle =\langle {\bf a}^\dagger_i\rangle=0, \quad \langle {\bf a}^\dagger_i{\bf a}^\dagger_j\rangle =
            \langle {\bf a}_i{\bf a}_j\rangle=0, \\
        & \langle {\bf a}^\dagger_i{\bf a}_j\rangle =\delta_{ij}n_j,
    \end{split}
\end{equation*}
with $n_j = ({e^{\Omega_j} - \zeta})^{-1}$, $\zeta = \pm 1$ for the bosonic and fermionic case, respectively. 
With these operators thus defined, the basis 
$B'=\{{\bf id}_{\hilbert}, {\bf a}_i, {\bf a}_i^\dagger, 
{\bf a}_i{\bf a}_j,
{\bf a}_i^\dagger{\bf a}_j^\dagger, 
{\bf a}_i^\dagger{\bf a}_j-\delta_{ij}n_i\}$
provides an orthogonal basis w.r.t. both the KMB and correlation scalar products.
These products differ only on the induced norm over the operators, as it is shown in table \cref{tbl:spcomp}.

Using this property, it is straightforward to span the projector as
$$
\pi({\bf O})=\sum_{{\bf Q}\in B} \frac{({\bf Q},{\bf O})^s}{({\bf Q},{\bf Q})^s}{\bf Q}
$$

\noindent where $s = {\rm KMB/covar}$ and where $\pi$ the orthogonal projector w.r.t the scalar product $(\cdot , \cdot)^s$.

\begin{table}
\begin{tabular}{c c c}
  \toprule
  &KMB & covar\\
  \hline
${\bf id}_{\hilbert}$ &1 & 1\\
${\bf a}_i$ &${1}/{\Omega_i}$& $n_i$ \\
${\bf a}_i {\bf a}_j$ &$\frac{1+n_i+n_j}{\Omega_i+\Omega_j}$& $n_i n_j + \frac{n_i + n_j + 1}{2}$ \\
${\bf a}^\dagger_i {\bf a}_j$& $\frac{n_j-n_i}{\Omega_i-\Omega_j}$ & $n_i n_j+ \frac{n_i+n_j}{2}$   \\
${\bf a}^\dagger_i {\bf a}_i-n_i$ & $n_i(n_i+1)$ &  $n_i(n_i+1)$\\
${\bf a}^2_i$&$\frac{1/2+n_i}{\Omega_i}$&$2 (n_i+1/2)^2+1/2$\\
\hline
\end{tabular}
\caption{\label{tbl:spcomp}Induced norms for the KMB and the correlation scalar products for the elements of the basis $B'$. Here, $i\neq j$, and the last line just applies for bosons.}
\end{table}

\subsection{MFA and Gaussian-state-based MFT as Max-Ent dynamics }\label{subsec:proof_mf_as_proj}

Standard mean-field treatment for composite quantum systems, both in the case of product-state-based and Gaussian-state-based versions, can be stated in terms of Max-Ent projections.
In the product-state case, the subspace ${\cal A}_B$ is defined by the local operators, in a way s.t.\ $B=\bigsqcup_i B_i$ wherein the $A_{B_i}$ sets define closed subalgebras of ${\cal A}$.
The Max-Ent states are, then, product states $\rho=\bigotimes_i \rho_i$. 
Moreover, general operators ${\cal O}\in {\cal A}$ can be written as linear combinations of products of local operators, with their expectation values written in terms of products of local expectation values. 
In this way, for product states, the expectation value of any observable is a functional of the expectation values of an independent set of local observables. 

On the other hand, for Gaussian-state-based MFT (both for the bosonic and the fermionic cases), ${\cal A}_{B}$ is the sub-algebra of quadratic forms in creation and annihilation operators, making the Max-Ent states Gaussian states.
Thanks to the Wick's theorem, expectation values can be written as linear combinations of products of expectation values of operators in $B$.

In both cases, the projection $\pi^\textnormal{MF}:{\cal A}\rightarrow {\cal A}_B$ can be written as
\begin{equation}\label{eq:MFProj}
  \pi_{\sigma}^{\textnormal{MF}}({\bf O})=\sum_{{\bf Q}\in B}({\bf Q}-\langle {\bf Q}\rangle_\sigma)\frac{\partial \langle  {\bf O}\rangle_{\sigma} }{\partial \langle {\bf Q}\rangle_{\sigma}}+\langle {\bf O}\rangle_{\sigma},
\end{equation}
with mean values evaluated regarding $\sigma \in {\cal M}_B$. The self-consistency equation for the stationary case can be written as
\begin{equation}
  \sigma = \frac{\exp(-\pi_{\sigma}^{\textnormal{MF}}({\bf H}))}{\Tr \exp(-\pi_{\sigma}^{\textnormal{MF}}({\bf H}))},
\end{equation}

while the time-dependent equations can be written as
$$
\frac{d{\bf K}}{dt}
   = \pi_{\sigma}^{{\textnormal{MF}}}\left(\frac{[{\bf H},{\bf K}]}{{\bf i}\hbar}\right).
   $$

We claim the following, 

\begin{prop}\label{prop:orthogMFproj}
    $\pi^\textnormal{MF}$ represents an \emph{orthogonal projection} regarding both the $\textnormal{KMB}$ and the correlation scalar product. 
\end{prop}

It is convenient, first, to consider some special basis of operators which simplifies the analytical evaluation of expansions and scalar products. 
In particular, for product state based MFT, $\sigma=\bigotimes_i \sigma_i$, we are going to use the local basis 

$$
    B_i=\bigg\{\bigg(|\alpha\rangle\langle\alpha'| - \langle\alpha'|\alpha\rangle \frac{{\bf id}_{{(i)}}}{\Tr \mathbf {id}_{{(i)}}}\bigg)\otimes {\bf id}_{\overline{i}}\bigg\},
$$ 

\noindent with $|\alpha\rangle$, $|\alpha'\rangle$ orthogonal eigenvectors of $\sigma_i$, ${\bf id}_i$ the identity operator on the subsystem $i$ and ${\bf id}_{\overline{i}}$ the identity operator over subsystem complementary to $i$. 
These operators are not all Hermitian.
However, since ${\bf Q}\in B_i \Leftrightarrow {\bf Q}^\dagger \in B_i$, it is possible to build an Hermitian basis by replacing ${\bf Q}$, ${\bf Q}^\dagger$ by their linear combinations ${\bf Q}_{\pm}=\frac{{\bf Q}\pm {\bf Q}^\dagger}{\sqrt{\pm 1}}$. 
Also, since we are interested in the connection with \emph{real}-valued scalar products, most of the results can be obtained from a restriction over the complexified version of ${\cal A}$ and ${\cal A}_B$. 
The main advantage of these bases is that, regarding $\sigma$,

   \begin{equation}\label{eq:sigmaconj2}
     \sigma^{\tau}{\bf Q}\sigma^{-\tau}=e^{\Omega_{\bf Q} \tau}{\bf Q},
   \end{equation}
with $\Omega_{{\bf Q} }=-{\Omega}_{{\bf Q}^\dagger}\in \mathbb{R}$. 

In a similar way, for Gaussian-state-based MFT, we are going to consider the basis $B=B_{\rm gaussian}$ generated by the identity ${\bf id}_\hilbert$, the canonical raising and lowering operators ${\bf a}_i^\dagger$, ${\bf a}_i$ and their pairwise products ($\{{\bf a}_i^\dagger{\bf a}_j, {\bf a}_i{\bf a}_j, {\bf a}_i^\dagger{\bf a}_j^\dagger\}$). 
With respect to the state $\rho_0\propto e^{-\sum_i\Omega_i {\bf a}_i^\dagger {\bf a}_i}$ these operators satisfy 
$$\langle {\bf a}_i\rangle=\langle {\bf a}^\dagger_i\rangle=\langle {\bf a}_i{\bf a}_j\rangle=\langle {\bf a}_i^\dagger{\bf a}^\dagger_j\rangle=
\langle {\bf a}_i^\dagger{\bf a}_j-\frac{\delta_{ij}}{e^{\Omega_{i}} - \zeta}{\bf id}_{\hilbert}\rangle=0$$
where $\zeta = \pm 1$ corresponds to bosonic (fermionic) statistics. 
This basis also generates the corresponding algebra ${\cal A}$, and satisfies \cref{eq:sigmaconj2}.

Regarding these bases, it is possible to prove the following

\begin{lem}\label{lemma:prodtoav}
  Let $({\bf Q},{\bf O})\equiv ({\bf Q},{\bf O})^\textnormal{KMB}$ or  $({\bf Q},{\bf O})\equiv ({\bf Q},{\bf O})^{\rm covar}$, and let ${\bf Q}\in B$ for $B_{\rm prod}$ or $B_{\rm sep}$ regarding the same state $\sigma$. Then, the following holds,
  
        \begin{equation}\label{ec:propspfromav}
            ({\bf Q},{\bf O})=({\bf Q},{\bf Q})\frac{\langle {\bf Q}^\dagger{\bf O}\rangle_{\sigma}}{\langle {\bf Q}^\dagger{\bf Q}\rangle_{\sigma}},
        \end{equation}
    \end{lem}
    \begin{proof}

  The property given in \cref{eq:sigmaconj2} verifies the following for the KMB product, 
  
  $$
  ({\bf Q},{\bf O})_{\sigma}^{\textnormal{KMB}}=\int_0^1 \Tr[\sigma^{1-\tau}{\bf Q}^\dagger\sigma^{\tau}{\bf O}]d\tau=
  \frac{e^{\Omega_{{\bf Q}}}-1}{\Omega_{{\bf Q}}}\langle {\bf Q}^\dagger{\bf O}\rangle_{\sigma},
  $$
  where the first factor in the RHS does not depend on ${\bf O}$. 
  Replacing this identity in \cref{ec:propspfromav} yields the equality. 

  In a similar way,
  \begin{eqnarray*}
    ({\bf Q},{\bf O})_{\sigma}^{\rm covar}&=&\Tr[\sigma \{{\bf Q}^\dagger,{\bf O}\}]/2\\
                                     &=&\Tr[\sigma {\bf Q}^\dagger{\bf O}+\sigma {\bf O}{\bf Q}^\dagger]/2\\
                                     &=&\Tr[\sigma {\bf Q}^\dagger{\bf O}+{\bf O}\sigma \sigma^{-1}{\bf Q}^\dagger\sigma]/2\\
                                         &=&\Tr[\sigma {\bf Q}^\dagger{\bf O}+e^{\Omega_{{\bf Q}}}{\bf O}\sigma {\bf Q}^\dagger]/2\\
                                             &=&\frac{e^{\Omega_{{\bf Q}}}+1}{2}\Tr[\sigma {\bf Q}^\dagger{\bf O}]=\frac{e^{\Omega_{{\bf Q}}}+1}{2}\langle{\bf Q}^\dagger{\bf O} \rangle_{\sigma}.\\
  \end{eqnarray*}
  \end{proof}

From the previous lemma, is easy to verify that $B_{\rm prod}$ and $B_{\rm gaussian}$ are orthogonal basis regarding the corresponding orthogonal products: the basis were chosen in a way that any pair of operators in $B$ are not correlated regarding the state $\sigma$.

To proceed with the proof of \cref{prop:orthogMFproj}, we are going to need the following two lemmas:

\begin{lem}\label{lemma:expectprod}
Be $\sigma=\bigotimes_i \sigma_i$, $B=B_{\rm prod}=\sqcup B_i$, ${\bf Q}\in B_i$ and ${\bf O}\in {\cal A}$. Then,
$$
\langle {\bf Q}^\dagger {\bf O}\rangle=\sum_{{\bf Q}'\in B} \langle {\bf Q}^\dagger {\bf Q}' \rangle\frac{\partial  \langle {\bf O}\rangle}{\partial \langle{\bf Q}'\rangle}.
$$
\end{lem}

\begin{proof}
Since any ${\bf O}\in{\cal A}$ can be expanded as a linear combination of products of operators in ${\bf B}$, and ${\bf Q}\in B_i$ for certain $i$, it is enough to prove the restriction to the case
${\bf O}={\bf O}_i {\bf O}_{\overline{i}}$, with ${\bf O}_i \in B_i$. Then,

\begin{eqnarray*}
\langle {\bf Q}^\dagger {\bf O} \rangle&=&\langle {\bf O}_{\overline{i}}\rangle
                                           \langle{\bf Q}^\dagger {\bf O}_i\rangle\\
 &=&\frac{
 \partial \langle {\bf O}_{\overline{i}}\rangle \langle {\bf O}_{i}\rangle}{\partial \langle {\bf O}_i\rangle}\langle{\bf Q}^\dagger {\bf O}_i\rangle\\
  &=& \frac{
      \partial \langle {\bf O}\rangle}{\partial \langle {\bf O}_i\rangle}\langle{\bf Q}^\dagger {\bf O}_i\rangle\\
  &=& \sum_{{\bf Q}'\in B}\frac{
      \partial \langle {\bf O}\rangle}{\partial \langle {\bf O}_i\rangle}
\frac{\partial \langle {\bf O}_i\rangle}{\partial \langle {\bf Q}'\rangle}
      \langle{\bf Q}^\dagger {\bf Q}'\rangle\\
&=& \sum_{{\bf Q}'\in B}\frac{
      \partial \langle {\bf O}\rangle}{\partial \langle {\bf Q}'\rangle}
      \langle{\bf Q}^\dagger {\bf Q}'\rangle,
\end{eqnarray*}
\end{proof}

\begin{lem}\label{lemma:expectgauss}
Be $\sigma=\bigotimes_i \sigma_i$, $B=B_{gauss}$, ${\bf Q}\in B$ and ${\bf O}\in {\cal A}$. Then,
$$
\langle {\bf Q}^\dagger {\bf O}\rangle=\sum_{{\bf Q}'\in B} \langle {\bf Q}^\dagger {\bf Q}' \rangle\frac{\partial  \langle {\bf O}\rangle}{\partial \langle{\bf Q}'\rangle}.
$$
\end{lem}

\begin{proof}
    This case follows a similar line that the proof of \cref{lemma:expectprod}, but based on the Wick's theorem~\cite{wick_evaluation_1950}. We start by assumig that ${\bf O}={\bf z}_1\ldots {\bf z}_n$
    and ${\bf Q}$ belongs to one of the following cases: 
    \begin{enumerate}
    \item ${\bf Q}={\bf w}_a$,
    \item ${\bf Q}={\bf w}_a{\bf w}_b-\langle {\bf w}_a{\bf w}_b\rangle$,
    \end{enumerate}
    with ${\bf z}_1,\ldots, {\bf z}_n$,${\bf w}_a$, ${\bf w}_b$ the elementary excitation operators ${\bf a}_i$, ${\bf a}_i^\dagger$. 
    Let's start by the first case. 
    \begin{eqnarray*}
        \langle {\bf Q}^\dagger {\bf O} \rangle&=&\sum_{k}\langle {\bf w}_a^\dagger {\bf z}_k \rangle 
        \langle {\bf z}_1\ldots \underline{{\bf z}_{k}}\ldots {\bf z}_{n}\rangle \\
        &=&\langle {\bf Q}^\dagger {\bf Q}  \rangle \frac{\partial \langle{\bf O}\rangle}{\partial \langle {\bf Q}\rangle} \\
        &=&\sum_{{\bf Q}'\in B}\langle {\bf Q}^\dagger {\bf Q}'  \rangle \frac{\partial \langle{\bf O}\rangle}{\partial \langle {\bf Q}'\rangle},
    \end{eqnarray*}
with $\underline{{\bf z}_k}$ meaning that the factor is removed from the product. To understand the last line, we notice first that $\langle {\bf w}_a^\dagger{\bf z}_k\rangle=0$ except for the case in which ${\bf z}_k\equiv {\bf w}_a$. Then, only these terms contribute to the sum. On the other hand, from the Wick's theorem, $\langle {\bf z}_1\ldots {\bf z}_{n}\rangle$ is a linear combination of products of the form $\langle z_{i_1}z_{i_2}\rangle\ldots \langle z_{i_{2 m-1}}z_{i_{2m}}\rangle \langle {\bf z}_{2m+1}\rangle \ldots \langle {\bf z}_{i_n}\rangle$ with $\{i_n\}$ a permutation over the original indices. Removing the operator ${\bf z}_k$ changes each of these terms by removing the corresponding 
$\langle{\bf z}_k\rangle$ factor, or by changing a $\langle {\bf z}_{k}{\bf z}_{i_m} \rangle$ by a factor proportional to $\langle{\bf z}_{i_m}\rangle$. The second change produces vanishing factors when are evaluated over $\sigma$, while the first just produces a finite contribution if there is just one factor $\langle {\bf w}_a\rangle$  in the product, which happens just when $n$ is an odd number. Finally, the last line follows from $\langle {\bf Q}^\dagger {\bf Q}'\rangle=0$ except for ${\bf Q}={\bf Q}'$.

In a similar way, the second case can be written as 
\begin{eqnarray*}
\langle {\bf Q}^\dagger {\bf O} \rangle&=&
\sum_{k<k'}\langle {\bf Q}^\dagger {\bf z}_i{\bf z}_j\rangle \langle {\bf z}_1 \ldots \overline{{\bf z}_k}\ldots \overline{{\bf z}_{k'}}\ldots {\bf z}_n\rangle,
\\
&=&\langle {\bf Q}^\dagger {\bf Q}\rangle \frac{\partial \langle {\bf O}\rangle}{\partial\langle {\bf Q} \rangle}\\
&=&\sum_{{\bf Q}'\in B}\langle {\bf Q}^\dagger {\bf Q}'  \rangle \frac{\partial \langle{\bf O}\rangle}{\partial \langle {\bf Q}'\rangle}.
\end{eqnarray*}

\end{proof}

\subsection{Proof of \cref{prop:orthogMFproj}}
Now we are in conditions to show the proof for \cref{prop:orthogMFproj}:
\begin{proof}
We start from the general condition for being $\pi$ an orthogonal projector regarding the scalar product $(\cdot,\cdot)$ is given by
$$
\bigg({\bf Q},\pi({\bf O})\bigg)=({\bf Q},{\bf O}).
$$
for any ${\bf Q}$ s.t.\ $\pi({\bf Q})={\bf Q}$. Replacing $\pi$ by $\pi^\textnormal{MF}$ and the scalar product with the KMB or the correlation scalar product, and using the result from \cref{lemma:prodtoav}, the condition reads
\begin{equation}\label{eq:condorthproj1}
\sum_{Q'\in B}\frac{\partial\langle{\bf O}\rangle}{\partial\langle{\bf Q}'\rangle}\langle {\bf Q}^\dagger {\bf Q'} \rangle=\langle {\bf Q}^\dagger {\bf O} \rangle.
\end{equation}
Using Lemmas \ref{lemma:expectprod} and \ref{lemma:expectgauss}, the RHS takes the same form as the LHS, which completes the proof.

\end{proof}

\end{appendix}
\bibliography{main.bbl}
\bibliographystyle{apsrev4-1}
\end{document}